\documentclass[11pt]{article}

\usepackage[dvipsnames]{xcolor}
\usepackage{graphicx}
\usepackage{amsmath}
\usepackage{amsthm}
\usepackage{amsmath}
\usepackage{amssymb}
\usepackage{amsfonts}
\usepackage{epsfig}
\usepackage{verbatim}
\usepackage{bm}
\usepackage{mathtools}
\usepackage{newtype}

\usepackage[colorlinks=true,pagebackref=false,breaklinks=true,citecolor=black,linkcolor=black,urlcolor=blue]{hyperref}
\usepackage[T1]{fontenc}

\newtheorem{theorem}{Theorem}[section]
\newtheorem{lemma}[theorem]{Lemma}
\newtheorem{corollary}[theorem]{Corollary}
\newtheorem{conjecture}[theorem]{Conjecture}
\newtheorem{proposition}[theorem]{Proposition}

\theoremstyle{definition}
\newtheorem{remarkx}[theorem]{Remark}
\newenvironment{remark}
  {\pushQED{\qed}\remarkx}
  {\popQED\endremarkx}

\theoremstyle{remark}
\newtheorem*{claim}{Claim}

\newcommand{\etal}{\textit{et al}.}

\newcommand{\NN}{\mathbb{N}}
\newcommand{\ZN}{\mathbb{Z}}
\newcommand{\QN}{\mathbb{Q}}
\newcommand{\RN}{\mathbb{R}}

\newtype{\class}{\mathbf}
\class{P}
\class{NP}
\class{AM}
\class{LOGNP}
\class{coNP}
\class{NTIME}
\class{PPAD}
\class{PSPACE}
\class{EXP}
\class{EXPSPACE}
\class{NEXP}
\class{RP}
\class{coRP}
\class{PP}
\class{PH}
\class{CH}
\class{FIXP}
\class[FIXPd]{FIXP_d}
\class[sharpP]{\#P}

\newcommand{\sh}[1]{\class<\Sigma^p_{#1}>}

\newcommand{\NPR}{\class<\ensuremath{\bm{\exists \RN}}>}
\newcommand{\coNPR}{\class<\ensuremath{\bm{\forall \RN}}>}

\newcommand{\PR}[2]{\class<\bm{\Pi}_{#1}^{#2} \bm{\RN}>}
\newcommand{\SR}[2]{\class<\bm{\Sigma}_{#1}^{#2} \bm{\RN}>}

\DeclareMathOperator{\relop}{\rho}

\DeclareMathOperator{\BP}{BP}

\DeclareMathOperator*{\opt}{opt}

\newcommand{\ERD}{\textsf{\upshape ERD}$_{\RN}$}
\newcommand{\LERD}{\textsf{\upshape LERD}$_{\RN}$}
\newcommand{\IED}{\textsf{\upshape ImageEDense}$_{\RN}$}
\newcommand{\LS}{\textsf{\upshape LocSupp}$_{\RN}$}

\newcommand{\FFEAS}{\ensuremath{\mbox{$4$-\textsf{\upshape FEAS}$_{\RN}$}}}

\newcommand{\STR}[2]{$\Sigma_{#1}^{#2}$-\textsf{\upshape TR}}
\newcommand{\PTR}[2]{$\Pi_{#1}^{#2}$-\textsf{\upshape TR}}
\newcommand{\SPOLY}[2]{$\Sigma_{#1}^{#2}$-\textsf{\upshape POLY}}
\newcommand{\PPOLY}[2]{$\Pi_{#1}^{#2}$-\textsf{\upshape POLY}}

\newcommand{\twocolext}{$2$-\textsf{\upshape COLORING EXTENSION}}

\newcommand{\co}{\ensuremath{\mathrm{co}\mbox{-}}}
\newcommand{\bo}{\ensuremath{\mathrm{bo}\mbox{-}}}
\newcommand{\bd}{\ensuremath{\mathrm{bc}\mbox{-}}}
\newcommand{\bowod}{\ensuremath{\mathrm{bo}}} %without dash
\newcommand{\bdwod}{\ensuremath{\mathrm{bc}}} %without dash

\def\CC{\cal C}
\def\RCC8{\textup{RCC8}}

\def\itemi{\item [$(i)$]}
\def\itemii{\item [$(ii)$]}

\begin{document}

\title{Beyond the Existential Theory of the Reals}

\author{
{Marcus Schaefer
} \\
{\small School of Computing} \\[-0.13cm]
{\small DePaul University} \\[-0.13cm]
{\small Chicago, Illinois 60604, USA} \\[-0.13cm]
{\small \tt mschaefer@cdm.depaul.edu}\\[-0.13cm]
\and
{Daniel \v{S}tefankovi\v{c}
} \\
{\small Computer Science Department} \\[-0.13cm]
{\small University of Rochester} \\[-0.13cm]
{\small Rochester, NY 14627-0226} \\[-0.13cm]
{\small \tt stefanko@cs.rochester.edu}\\[-0.13cm]
}

\date{}

\maketitle

\begin{abstract}
 We show that completeness at higher levels of the theory of the reals is a robust notion (under changing
 the signature and bounding the domain of the quantifiers). This mends recognized gaps in the hierarchy, and
 leads to stronger completeness results for various computational problems. We exhibit several families of complete problems which can be used for future completeness results in the real hierarchy. As an application we sharpen some results by B\"{u}rgisser and Cucker on the complexity of properties of semialgebraic sets, including the Hausdorff distance problem also studied by Jungeblut, Kleist, and Miltzow.
\end{abstract}

{\bfseries\noindent Keywords: }
existential theory of the real numbers, theory of the reals, real hierarchy, computational complexity, semialgebraic sets, Hausdorff distance.

{\bfseries\noindent MSC Classification: }
68Q15, 68Q17, 14P10

\maketitle

\section{Introduction}

Three recent papers have captured the complexity of computational problems at the second level of the theory of the real numbers: Jungeblut, Kleist and Miltzow~\cite{JKM22} show that computing the Hausdorff distance of two semialgebraic sets is as hard as the $\forall\exists_{<}$-fragment of the theory of the reals (with $<$ as the only relation, no negation); D'Costa, Lefaucheux, Neumann, Ouaknine, and Worrell~\cite{DCLNOW21} study a problem in linear dynamical systems, whose complement is also as hard as the $\forall\exists_{<}$-fragment; and Dobbins, Kleist, Miltzow and Rz\k{a}\.{z}ewski~\cite{DKMR18,DKMR22} showed that a version of area universality is as hard as the $\forall\exists$-fragment (with no restriction on relations). Moreover, Blanc and Hansen~\cite{BH22a} locate a problem in evolutionary game theory within the $\forall\exists$-fragment (without proving it hard).

The different restrictions on which relations are allowed suggest the possibility that there are distinct sublevels of complexity at the second level of the theory of the reals. We show in Section~\ref{sec:HRC} that this is not the case: there is only one such level of complexity. Previously, such a result was known only for the first level of the hierarchy~\cite{SS17}, justifying the complexity class \NPR, which corresponds to the existential theory of the reals. The robustness results of this paper justify the {\em real (complexity) hierarchy} corresponding to all finite levels of the theory of the reals. We
identify several (very restricted) families of complete problems for the real hierarchy.
These problems should be useful in future hardness proofs.

As a first step towards testing that claim, Section~\ref{sec:SAS} looks at the computational complexity of various problems about semialgebraic sets. We extend and sharpen relevant work by B\"{u}rgisser and Cucker~\cite{BC09} who first determined the complexity of these problems in the Blum-Shub-Smale model of real computation. The descriptions of some of these problems require exotic (non-standard) quantifiers; we will show that when determining the complexity of these problems in the real hierarchy, exotic quantifiers can be eliminated in nearly all the natural problems considered by B\"{u}rgisser and Cucker.

Before we can state our results formally, we introduce the theory of the reals in some more depth, describe the encoding of computational problems, and discuss how our work relates to the Blum-Shub-Smale model of real computation.

\begin{remark}[Changes and Corrections]\label{rem:changes}
Towards the end of 2024 it was brought to our attention that proofs of results in Section~\ref{sec:BU} on the bounded universe case were missing important details~\cite{HP24}.
%Towards the end of 2024 Kristoffer Arnsfelt Hansen and Soumyajitit Paul alerted us that proofs of results in Section~\ref{sec:BU} on the bounded universe case were missing important details~\cite{HP24}.
On review we found that the details for the bounded-open case could be fixed, this case is now treated in Section~\ref{sec:bouni}. On the other hand, we were not able to repair all the proofs for bounded-closed domains; we now have to admit the possibility that classes over restricted signatures may be weaker in this case. We do not believe that this is really true, but we do not have a proof that it is not. We were able to show that the bounded-domain versions do not proliferate. This new material is presented in Section~\ref{sec:bduni}. We have taken the opportunity of this rewrite to strengthen some of the degree bounds we claimed in the original version.

Changes in more detail: {\em Corollary~2.10} of the original version had to be restricted to the open case, this is now Corollary~\ref{cor:strictify2}. The bounded-open claims of {\em Proposition~2.11} and {\em Corollary 2.12} of the original paper are now covered by Proposition~\ref{prop:boundedopen}. The claims for bounded-closed domains in {\em Proposition~2.11} and {\em Corollary 2.12} have become open questions. Instead we have the weaker Proposition~\ref{prop:bcunb} on unrestricted signatures, and the non-proliferation result~\ref{prop:bceq}. For the first and second levels, the bounded-domain case can be resolved, see Proposition~\ref{prop:bc2nd}, but for higher-levels it remains open, see Conjecture~\ref{con:bceq}.

Independently of this, we also added some detail to {\em Lemma~1.7} of the original paper, which is now Lemma~\ref{lem:Tseitin} to improve some of the later degree bounds, in particular the bound in Proposition~\ref{prop:QPOLY} (same numbering in both versions).
\end{remark}

\subsection*{The Theory of the Reals}

The {\em theory of the reals} is the set of all sentences (that is, no free variables) true over the real numbers. It includes statements such as $(\exists x)\ x^ 2 > 0$, $(\forall y)(\exists x)\ x^2 = y \vee x^2 = -y$, and $(\exists x)(\forall y_1, y_2)(\exists z)\ x+y_1 = y_1 + x \wedge y_2+z = x$. Tarski~\cite{T48} showed that the theory of the reals, and thereby analytic geometry, is decidable; for a modern treatment, see~\cite{BPR06}.

The theory of the reals is highly expressive. Many natural problems in computational geometry, graph drawing and other areas can be expressed in small fragments of the theory, most prominently, the existential fragment, in which only existential quantifiers are allowed. Within the past decade it has become increasingly evident that the existential fragment not only expresses many of these problems, but captures their computational, and sometimes algebraic, complexity precisely. This led to the introduction of a complexity class, \NPR\ (we read this as ``exists R''), capturing the complexity of deciding truth in the {\em existential theory of the reals}. This is analogous to \NP, which captures the complexity of deciding truth of existentially quantified statements over $\{0,1\}$. Recent examples of problems shown to be \NPR-complete include
the art gallery problem~\cite{AAM22},
polygon coverings~\cite{A22},
angular resolution of a graph~\cite{S23},
continuous constraint satisfaction problems~\cite{MS22},
Nash equilibria~\cite{BH22},
and training neural networks~\cite{BHJMW22}. For a very enjoyable and thorough introduction to the existential theory of the reals, see Matou{\v{s}}ek~\cite{M14}; for a (partial) list of \NPR-complete problems, see~\cite{W12}.

Just as in the case of \NP, attention has extended beyond the existential fragment. Many computational problems very naturally involve quantifier alternation. Take, as an example, the notion of Hausdorff distance. Given two sets $A, B \subseteq \RN^n$, the {\em directed Hausdorff distance from $A$ to $B$} is $\sup_{x\in A} d(x,B)$, where $d(x,B) = \inf_{y \in B} \lVert x-y \rVert$ and $\lVert \cdot \rVert$ is the Euclidean norm. The {\em Hausdorff distance}, $d_H(A,B)$, of two sets $A$ and $B$ is the minimum of the directed Hausdorff distances from $A$ to $B$ and from $B$ to $A$. If $A$ and $B$ are semialgebraic, membership in $A$ and $B$ is encoded by Boolean formulas $\varphi(x)$ and $\psi(y)$ in the quantifier-free fragment, so $d_H(A,B) \leq r$ can be expressed as
\begin{equation*}
\begin{split}
(\forall \varepsilon > 0) & (\forall a,b)(\exists a',b') \\
                         & \left(\varphi(a) \wedge \psi(b)\right) \implies \left(\psi(b') \wedge \varphi(a') \wedge \max(\lVert a-b'\rVert,\lVert b-a'\rVert) < r+\varepsilon\right),
\end{split}
\end{equation*}
which belongs to the $\forall\exists$-fragment of the theory of the reals.
But is it as hard as deciding truth in that fragment?

That questions turns out to be somewhat subtle. It is straightforward to define complexity classes $\SR{k}{}$ and $\PR{k}{}$ that correspond to deciding the truth in the $\exists\forall\cdots$ and $\forall\exists\cdots$ ($k-1$ quantifier alternations) fragments of the theory, but showing hardness for these classes has run into issues. What happened?

The general theory of the reals allows a large signature, including function symbols $\{0,1,+,\cdot\}$ and relation symbols $\{<, \leq, =, >, \geq\}$.\footnote{We do allow abbreviations, e.g.\ we write powers such as $x^4$, but we understand that to be shorthand for $x\cdot x \cdot x \cdot x$.} Clearly, the expressiveness of the language does not change if we drop $>$ and $\geq$.\footnote{We will keep using both symbols, since they can easily be expressed by exchanging lhs and rhs of the inequality and flipping the sign.}
What happens with $=$, $<$ and $\leq$ though?
To make the question non-trivial, we have to disallow negation; because of trichotomy, $x<y \vee x = y \vee y<x$, that is not an onerous restriction.

There are seven non-trivial combinations of $\{<, \leq, =\}$, but many of these are easily seen to be equivalent (in terms of expressiveness). For example, once we have $<$ and $=$, we have $\leq$, using an ``or'', and once we have $\leq$, we have $=$, using an ``and''. It follows that there are only four signatures of interest: $\{<\}$, $\{\leq\}$, $\{=\}$ and $\{<,\leq, =\}$.

When studying the existential theory of the reals in~\cite{SS17}, we identified only two of these:
$\NPR_{<}$, corresponding to $\{<\}$ and
$\NPR$, corresponding to $\{<,\leq, =\}$. There is a reason for that: with an existential quantifier we can
define $a < b$ as $(\exists x)\ ax^2+1 \leq bx^2$, and $a \leq b$ as $(\exists x)\ a + x^2 = b$, so in the existential theory, only signatures $\{<\}$ and $\{<,\leq, =\}$ are of interest. Since
$\NPR_{<}= \NPR$, as we showed in~\cite{SS17}, \NPR\ is very robust as a complexity class: Problems identified earlier as $\NPR_{<}$-complete, such as rectilinear crossing number~\cite{B91} and segment intersection graphs~\cite{KM94}, and problems identified as $\NPR$-complete, such as the pseudoline stretchability problem~\cite{M88,S91b} are all computationally equivalent (even if they differ algebraically).

As researchers began exploring the second level of the theory, similar gaps became apparent. D'Costa, Lefaucheux, Neumann, Ouaknine, and Worrell~\cite{DCLNOW21} suggest three variants of each class:
\SR{k}{}, \SR{k}{<}, and $\SR{k}{\leq}$. They observe that $\SR{k}{} = \SR{k}{<}$ for $k=2$, since any equality $a = b$ can be replaced with $(\forall x)\  (x(a-b)) ^2 < 1$, and the universal quantifier can be added to the final, universal, block of quantifiers; clearly the same argument applies for all even $k \geq 2$, and to $\PR{k}{}$ if $k\geq 1$ is odd, e.g.\ for $\forall\RN$. D'Costa \etal\ did not consider the variant $\SR{k}{=}$.

In this terminology, D'Costa, Lefaucheux, Neumann, Ouaknine, and Worrell~\cite{DCLNOW21} showed that the compact escape problem is complete for $\SR{2}{\leq}$ (which implies that the complement of the problem is $\PR{2}{<}$-complete). Dobbins, Kleist, Miltzow and Rz\k{a}\.{z}ewski~\cite{DKMR18,DKMR22} showed that a variant of area universality is $\PR{2}{}$-complete. And
Jungeblut, Kleist and Miltzow~\cite{JKM22} showed that the Hausdorff distance of two semialgebraic sets is
$\PR{2}{<}$-complete. Earlier, B\"{u}rgisser and Cucker~\cite{BC09} showed that testing whether a rational function is surjective is $\PR{2}{}$-complete.

\smallskip

Our main contribution in this paper is to show that all these variants, with one exception, are the same, not just at the second level, but at all finite levels.

\begin{theorem}\label{thm:strictify}
  We have $\SR{k}{} = \SR{k}{<}$ and $\PR{k}{} = \PR{k}{<}$ for all $k \geq 1$.
\end{theorem}

It follows that computing the Hausdorff distance of two semialgebraic sets
is \PR{2}{}-complete~\cite{JKM22}. Theorem~\ref{thm:strictify} resolves an open question by
D'Costa, Lefaucheux, Neumann, Ouaknine, and Worrell~\cite{DCLNOW21}.

\begin{corollary}\label{cor:leq}
  We have $\SR{k}{} = \SR{k}{\leq}$ and $\PR{k}{} = \PR{k}{\leq}$ for all $k \geq 1$.
\end{corollary}

The corollary implies that the compact escape problem, as defined in~\cite{DCLNOW21}, is \SR{2}{}-complete.

With an eye towards future hardness reductions, we also investigate bounded universes. Specifically, we consider $I = (-1,1)$, the {\em bounded open} case and $I = [-1,1]$ the {\em bounded closed} case. We will use prefixes ``\bowod'' and ``\bdwod'' to distinguish these variants of the complexity classes.

\begin{corollary}\label{cor:bound}
  We have $\SR{k}{} =  \bo\SR{k}{<}$ and $\PR{k}{} = \bo\PR{k}{<}$ for all $k \geq 1$.
\end{corollary}

As in the unbounded case it also follows that $\SR{k}{} =  \bo\SR{k}{\leq}$ and $\PR{k}{} = \bo\PR{k}{\leq}$ for all $k \geq 1$.

For the bounded-close case we can only show that $\SR{k}{} =  \bd\SR{k}{}$ and $\PR{k}{} =  \bd\PR{k}{}$ for all $k \geq 1$, see Proposition~\ref{prop:bcunb}; for the restricted signature cases we can show that $\bd\SR{k}{<} = \bd\SR{k}{=}$ for odd $k
\geq 1$, see Proposition~\ref{prop:bceq}, but not that they are the same as $\SR{k}{}$. This remains a striking open problem.

\begin{conjecture}\label{con:bceq}
 $\SR{k}{} =  \bd\SR{k}{<}$ for all $k \geq 1$.
\end{conjecture}

For $k = 1$ the conjecture is known to be true by a result due to Ouaknine and Worrell~\cite[Theorem 7]{OW17}. For $k=2$ we show in Proposition~\ref{prop:bc2nd} that the truth of the conjecture follows from work by Jungeblut, Kleist and Miltzow~\cite{JKM24}.
The bounded-domain results generalize earlier work by the authors~\cite{SS17} for $k = 1$, and D'Costa, Lefaucheux, Neumann, Ouaknine, and Worrell~\cite{DCLNOW21} for $k = 2$.

\smallskip

Theorem~\ref{thm:strictify}, as well as Corollaries~\ref{cor:leq} and~\ref{cor:bound} can be sharpened by exhibiting very restricted families of sentences which are $\SR{k}{}$-complete. This is done in Propositions~\ref{prop:QPOLY}, \ref{prop:strictcomp}, and~\ref{prop:bceq}; Tables~\ref{tab:shcomp} and~\ref{tab:phcomp} summarize the families of complete problems.

\begin{remark}[Equality]\label{rem:eq}
Noticeably absent are $\SR{k}{=}$ and $\PR{k}{=}$ because they behave anomalously. If the final quantifier block is existential we can define $\leq$ and $<$ using $a\leq b \Leftrightarrow (\exists x)\ b=a+x^2$ and
$a<b \Leftrightarrow (\exists x,x')\ b=a+x^2\wedge xx' = 1$. So $\SR{k}{=} = \SR{k}{}$ for all odd $k$ and
$\PR{k}{=} = \PR{k}{}$ for all even $k$; this was first shown by Dobbins, Kleist, Miltzow and Rz\k{a}\.{z}ewski~\cite[Theorem 2.2]{DKMR22}.

On the other hand, if the final quantifier block is universal, equality is weak. Using the Schwartz-Zippel lemma one can easily show that $\PR{1}{=}\subseteq \coRP$, so each language in $\PR{1}{=}$ can be decided in randomized polynomial time, and it is quite possible that $\PR{1}{=} = \P$. The complexity of \PR{1}{=}\ is closely related to polynomial identity testing, a notoriously difficult field, see~\cite{SV15} for a paper investigating a similar model.

It may be true that $\SR{k}{=} = \SR{k-1}{}$ for even $k \geq 2$ and $\PR{k}{=} = \PR{k-1}{}$ for odd $k \geq 3$, and
this may be easier to settle than the case of \PR{1}{=}\ by making use of the additional quantifiers.
\qed\end{remark}

\begin{remark}[Addition only]
 Sontag~\cite{S85} showed that if we do not allow multiplication, the resulting complexity classes are exactly the levels of the polynomial-time hierarchy. So anybody conjecturing $\NP = \NPR$ or $\sh{k} = \SR{k}{}$ is conjecturing that multiplication can be eliminated at these levels.
\qed\end{remark}

\subsection*{A Word on Encodings}

Since we are working within a logical fragment, some natural objects require encoding. The theory of the reals allows only two constants: $0$ and $1$, so we have to write integers using only these constants; computing the number in binary using Horner's scheme, this can be done efficiently, e.g. $13 = (1+1)\cdot ((1+1)\cdot ((1+1)+1))+1$. So we will write integer constants, since we know that they can be removed without increasing the encoding length significantly. We have to be more careful with polynomials. For us, a (multivariate) polynomial is a sum of monomial terms, such as $107x_1^3x_3^5-93x_2x_3^4$. In formulas we allow arithmetical terms, like $(x_1-x_2)(x_1-x_3)(x_2-x_3)$, but in a polynomial representation, we would have to multiply out the terms, which can lead to an exponential blow-up in the encoding length. With the help of existential quantifiers, such a blow-up can be avoided using a standard Tseitin-reduction. We will use the following lemma, which is based on Lemma 3.2 in~\cite{SS17}.

\begin{lemma}\label{lem:Tseitin}
   From a Boolean formula $\varphi(x_1, \ldots, x_i)$ in real variables $x_1, \ldots, x_i$ we can efficiently construct a
   family of quadratic polynomials $f_{\ell}:\RN^{i+j}\rightarrow \RN$, $\ell \in [k]$ such that
   \[\varphi(x_1, \ldots, x_i) \Leftrightarrow (\exists y_1, \ldots y_j) \bigwedge_{\ell \in [k]} f_{\ell} (x_1, \ldots, x_i, y_1, \ldots, y_j) = 0\]
  for all $x_1, \ldots, x_i$; each $f_{\ell}$ depends on at most four variables. %The parameters $j$ and $k$ are polynomial in the length of $\varphi$.
  Defining $f = \sum_{\ell \in [k]}f^2_{\ell}$, we obtain a non-negative polynomial $f: \RN^{i+j}\rightarrow \RN$ of degree at most $4$
  such that
  \[\varphi(x_1, \ldots, x_i) \Leftrightarrow (\exists y_1, \ldots, y_j)\ f(x_1, \ldots, x_i, y_1, \ldots, y_j) = 0\]
  for all $x_1, \ldots, x_i$.
\end{lemma}

We study computational problems over different types of sets defined by conditions on polynomials. A set $S$ is
\begin{itemize}
\item {\em algebraic} if $S = \{x: f(x) = 0\}$ for some polynomial $f$,
\item {\em basic semialgebraic} if $S = \{x: f(x) = 0, g_1(x) \geq 0, \ldots, g_n(x) \geq 0, h_1(x) > 0, \ldots, h_m(x) > 0\}$ for some polynomials $f$, $g_i$, and $h_j$, and
\item {\em semialgebraic} if it is the Boolean combination of basic semialgebraic sets, equivalently, if there is a Boolean formula $\varphi(x)$ consisting of conditions on polynomials such that $S = \{x: \varphi(x)\}$.
\end{itemize}

We specify algebraic and basic semialgebraic sets by their polynomials. A semialgebraic set is given by its defining formula $\varphi$. Lemma~\ref{lem:Tseitin} implies that every semialgebraic set can be viewed as the projection of an algebraic set, and semialgebraic sets are sometimes specified that way: $S = \{x: (\exists y)\ f(x,y) = 0\}$, though we will not do so here. Semialgebraic sets can also be specified as $S = \{x: C(x) = 1\}$, where $C$ is an {\em algebraic circuit}, see~\cite[Section 2]{BC09} for a precise definition. Circuits are a very succinct way of encoding a semialgebraic set central to the BSS-model, which we discuss next. It is unlikely that algebraic circuits can be turned into formulas without significant overhead, or the addition of existential quantifiers.

\subsection*{The Blum-Shub-Smale Model}

There is a notion of Turing-machine computations over real (and complex) numbers introduced by Blum, Shub, and Smale, known as the {\em BSS-model}~\cite{BCSS98}; Turing-machines in this model can have real constants, and perform addition, multiplication and comparisons of arbitrary real numbers in a single step. In the BSS-model, a real Turing machine takes as an input a tuple of real numbers and accepts or rejects if it terminates. A real Turing machine accepts a subset of $\RN^* = \bigcup_k \RN^k$. In this way, we can extend computation over $\{0,1\}$ to $\RN$, or any ring.

As an example, consider the problem \FFEAS\ which asks whether a polynomial of degree at most $4$ is {\em feasible}, that is, has a zero. A polynomial of degree at most $4$ can be written $f(x) = \sum_{\alpha} c_{\alpha} x^{\alpha}$, where $\alpha$ ranges over $\{0,1,2,3,4\}^n$, and each $\alpha$ has weight at most $4$. In the BSS-model we then encode $f$ as the sequence $c_{\alpha}$ of its $O(n^4)$ coefficients.

In analogy with the discrete classes, Blum, Shub and Smale introduced real equivalents of $\P$ and $\NP$, known as $\P_{\RN}$ and $\NP_{\RN}$. Given a polynomial $f$, a real Turing-machine with non-determinism can guess $x =(x_1, \ldots, x_n) \in \RN^n$ and verify that $f(x) = 0$. This shows that $\FFEAS \in \NP_{\RN}$. Blum, Shub, and Smale~\cite[Section 6]{BSS89} proved an analogue of Cook's theorem for $\NP_{\RN}$ by showing that the \FFEAS-problem is $\NP_{\RN}$-complete. Since then, these classes, with their own, real ``$\P_{\RN} = \NP_{\RN}?$''-problem have been studied in depth.

From the real world, we can return to discrete computational problems in two steps: we first restrict the real Turing-machines so they are only allowed to work with real constants $0$ and $1$. We write  $\P^0_{\RN}$ and $\NP^0_{\RN}$ for the resulting restricted classes. These classes are still classes of real numbers, so in a second step we discretize a class $\CC$, by intersecting it with $\{0,1\}^* = \bigcup_k \{0,1\}^k$, resulting in the {\em Boolean part}, $\BP(\CC)$, of $\CC$. We also write $\BP^0(\CC)$ for $\BP({\CC}^0)$, that is, applying both discretization steps to the class $\CC$ at once.

B{\"u}rgisser and Cucker~\cite[Proposition 8.5]{BC06} showed that the discrete version of the \FFEAS-problem: decide the feasibility of a polynomial of degree at most $4$ with integer coefficients, is  $\BP(\NP^0_{\RN})$-complete. Since the same problem is \NPR-complete~\cite[Theorem 4.1]{SS17}, this implies that
$\BP(\NP^0_{\RN}) = \NPR$, so the discretization of $\NP_{\RN}$ corresponds exactly to \NPR.

Blum, Cucker, Shub, and Smale~\cite[Section 21.4]{BCSS98} also introduce the notion of a {\em (real) polynomial hierarchy}, and they identify a family of complete problems (see comment after Proposition~\ref{prop:QPOLY}). B{\"u}rgisser and Cucker~\cite[Section 9]{BC09} introduced
a discrete version of that hierarchy, e.g.\ writing $\BP^0(\forall\exists)$ for the second level of that hierarchy. They identify a complete problem for each level, which is essentially the same problem we are working with, see Proposition~\ref{prop:QPOLY} below. This shows that the discretized real polynomial hierarchy, and the hierarchy we introduce below are (extensionally) the same. In particular $\BP^0(\forall\exists) = \PR{2}{}$.

Intensionally, though, our two approaches differ. B{\"u}rgisser and Cucker derive the discretized problems from their real counterparts in the real Turing machine model. We start with discrete computational problems and focus on determining their complexity through their logical structure. This gives us a more fine-grained view, closer to descriptive complexity, which can distinguish variant complexity classes like $\PR{k}{<}$, $\PR{k}{}$, and $\bd\PR{k}{}$, while proving that they are equal. Section~\ref{sec:SAS} contains a more detailed discussion of results from~\cite{BC09}.

Basu and Zell~\cite{BZ10} prove a real analogue of Toda's theorem for the BSS polynomial hierarchy by showing that computing the Betti numbers of a semialgebraic set is hard for that hierarchy.

\section{A Hierarchy of Real Complexity}\label{sec:HRC}

In Section~\ref{sec:LTR} we define complexity classes corresponding to various fragments of the theory of the reals,
establishing a {\em real (complexity) hierarchy}, in analogy with the polynomial-time hierarchy. Our main contribution is to show that the levels of the real hierarchy are robust under certain definitional modifications. We base our results on an effective {\L}ojasiewicz inequality due to Solern\'{o}~\cite{S91}, as explained in Section~\ref{sec:BG}. In Section~\ref{sec:BU} we extend our results to bounded quantification.

\subsection{Leveling the Theory of the Reals}\label{sec:LTR}

We assume sentences are in {\em prenex} form, that is,
\[(Q_1 x_1)(Q_2 x_2) \cdots (Q_i x_i)\ \varphi(x_1, \ldots, x_i),\]
where $Q \in \{\exists,\forall\}$ and $\varphi$ is a quantifier-free Boolean formula (without negation) over atomic predicates of the form $t(x_1, \ldots, x_i) \relop 0$, where $\relop \in \{<, \leq, =, \geq, >\}$ and $t(x_1, \ldots, x_i)$ is an arithmetical term in the (real) variables $x_1, \ldots, x_i$; $\varphi$ is known as the {\em matrix} of the sentence. We call a sentence {\em strict} if it does not contain the predicates $\leq$, $=$, and $\geq$ (and this is the reason we do not allow negation; with it, these predicates can be simulated).

We say a sentence is in $\Sigma_k$ $(\Pi_k$) if $Q_1 = \exists$ ($Q_1 = \forall$)
and there are $k-1$ quantifier alternations. We let \STR{k}{}\ (\STR{k}{<}) be the set of all true
(strict) sentences of the form $\Sigma_k$, and, similarly for \PTR{k}{}. Taking these as complete problems, we define
a hierarchy of complexity classes: $\SR{k}{}$ ($\SR{k}{<}$) is the set of problems that polynomial-time reduce to
\STR{k}{}\ (\STR{k}{<}) and similarly for \PR{k}{}\ and \PR{k}{<}.

For small, finite $k$, there are alternative names for these classes, including $\SR{1}{} = \NPR$, $\PR{1}{} = \coNPR$, $\PR{2}{} = \forall\exists\RN$, $\SR{2}{} = \exists\forall\RN$; for the strict versions we have $\exists_{<}\RN$, $\forall\exists_<\RN$ and so on.

In terms of traditional complexity, Renegar~\cite{R92a,R92b,R92c} showed that $\SR{k}{}$ and $\PR{k}{}$ can be decided in time $O(2^{2^{O(k)} n^k})$, where $n$ is the length of the formula. So $\SR{k}{}$ and $\PR{k}{}$ lie in $\EXP$ for any fixed $k$ (they even lie in \PSPACE, see~\cite[Remark 13.10]{BPR06}).

By definition, $\SR{k}{<} \subseteq \SR{k}{}$ and $\PR{k}{<}\subseteq \PR{k}{}$. In Section~\ref{sec:BG} we will show that equality holds for all $k$. This closes the gap between two possible definitions of the $k$-th level of the real hierarchy. We take the investigation further, by sharpening the form of the complete problems for \STR{k}{}\ and \PTR{k}{}, restricting the matrix of the sentences. Let \SPOLY{k}{\relop}\ (and \PPOLY{k}{\relop}) be the special case of \STR{k}{}\ (\PTR{k}{}) in which the matrix of the sentences has the form $\varphi(x_1, \ldots, x_i) = f(x_1, \ldots, x_i) \relop 0$ for a polynomial $f \in \ZN(x_1, \ldots, x_i)$ and $\relop \in \{=, <, \leq\}$. If we specify $\relop = \#$, we assume $\relop$ to be $=$ if the final quantifier block of the sentence is existential, and $<$ otherwise.

\SPOLY{k}{\#}\ and \PPOLY{k}{\#}\ are complete for \SR{k}{}\ and \PR{k}{}. This was first shown by B\"{u}rgisser and Cucker~\cite[Section 9]{BC09}, where the problem (with $\neq$ in place of $<$) is called $\textsf{STANDARD}^{\ZN}(Q_1\cdots Q_k)$; the same result was also observed in~\cite{DCLNOW21}. In the BSS-model, the result goes even farther back, to Blum, Cucker, Shub, and Smale~\cite[Section 21.4]{BCSS98}.
%; they avoid the degree increase to $8$ by allowing $\neq$ as a comparison operator.

\begin{proposition}\label{prop:QPOLY}
 \SPOLY{k}{\#}\ is \SR{k}{}-complete, and \PPOLY{k}{\#}\ is \PR{k}{}-complete for all $k \geq 1$. We can assume that polynomial in the matrix is non-negative and of degree at most $4$, more specifically, it is the sum of squares of quadratic polynomials, each depending on at most four variables. % (at most $4$ if the final quantifier block is existential).
\end{proposition}

\begin{proof}
If the final quantifier block is existential, we use a Tseitin reduction, Lemma~\ref{lem:Tseitin}, to find a polynomial $f \in \ZN(x_1, \ldots, x_i, y_1, \ldots, y_j)$ such that
 $\varphi(x_1, \ldots, x_i)$ is equivalent to $(\exists y_1, \ldots, y_i)\ f(x_1, \ldots, x_i, y_1, \ldots, y_j) = 0$.
The additional existential quantifier can be absorbed by the final quantifier block of the original sentence, since it is existential as well.

Otherwise, the final quantifier block is universal. We apply the Tseitin reduction to $\neg \varphi(x_1, \ldots, x_i)$ to get a non-negative polynomial $f \in \ZN(x_1, \ldots, x_i, y_1, \ldots, y_j)$ such that $\varphi(x_1, \ldots, x_i)$ is equivalent to $(\forall y_1, \ldots, y_i)\ f(x_1, \ldots, x_i, y_1, \ldots, y_j) > 0$, using non-negativity of $f$, and the additional universal quantifiers can be absorbed into the final quantifier block.
%Since $f$ is non-negative, this is eq \neq 0$ is equivalent to $f^2 > 0$ the result follows in this case as well. Computing $f^2$ can at most square the encoding length and double the degree.
\end{proof}

Proposition~\ref{prop:QPOLY} directly implies Theorem~\ref{thm:strictify} for those cases in which the final quantifier block is universal, for example, $\forall_<\RN = \forall\RN$ and $\exists\forall_<\RN = \exists\forall\RN$.

\subsection{Bridging the Gap}\label{sec:BG}

We work with an effective {\L}ojasiewicz inequality due to Solern\'{o}~\cite[Theorem 3]{S91}. Write $[n]$ for $\{1,\ldots,n\}$.

\begin{theorem}[Solern\'{o}]\label{thm:Linq}
There exists a positive integer constant $c$ such that the following is true. Let $\Omega=\prod_{j\in[n]} [s_j,t_j]\subseteq{\mathbb R}^n$.
Let $f,g\in\RN[x_1,\dots,x_n]$ be polynomials of total degree at most $D$. Assume that $f(x_1,\dots,x_n)=0$ implies $g(x_1,\dots,x_n)=0$ for all $x_1,\dots,x_n\in\Omega$.
Then there exists $\alpha>0$ (dependent on $f,g$, and $\Omega$) such that
\[
\vert f(x_1,\dots,x_n)\vert \geq \alpha \vert g(x_1,\dots,x_n)\vert ^{L}
\]
for all $x_1,\dots,x_n\in\Omega$, where $L=D^{n^c}$. If $f,g \in \ZN[x_1, \ldots, x_n]$, then
we can choose $\alpha = 2^{-\ell D^{cn^2}}$, where $\ell$ is the number of bits in the longest coefficient of $f$ and $g$.
\end{theorem}

\begin{remark}
The reader familiar with the proof that $\NPR = \NPR_{<}$ in~\cite{SS17} may remember that it was based on
an explicit lower bound on positive polynomials by Jeronimo and Perrucci~\cite{JP10}. Why are we taking recourse to
Solern\'{o}'s much older, less explicit bound~\cite{S91} here? The reason is that Solern\'{o}'s paper is one of the last that gives bounds for polynomials with real coefficients; more recent bounds tend to assume integer coefficients, and express bounds in terms of the bit-length of the coefficients, rather than---as Solern\'{o} did---their size. While it is possible that these newer, more explicit bounds can be extended to real coefficients, this would probably require reproving these results carefully. Since we do not require explicit bounds for our results, Solern\'{o}'s effective {\L}ojasiewicz inequality is fully sufficient for our purposes.
\qed\end{remark}

Theorem~\ref{thm:Linq} allows us to turn an $\exists_=$ into a $\forall\exists_{<}$ statement.

\begin{theorem}\label{thm:strictify2}
Let $f(x_1,\dots,x_n)\in \RN[x_1,\dots,x_n]$ be of total degree at most $D \geq 2$.
Then the sentence
\[
\Phi = (\exists x_1,\dots,x_n)\ f(x_1,\dots,x_n) = 0,
\]
is equivalent to
\[
\begin{split}
\Psi = (\forall z>0) & (\exists y\in (0,z)) (\exists x_1,\dots,x_n)\ \\
                     & \Big(\vert f(x_1,\dots,x_n)\vert < y^C\Big)\wedge \bigwedge_{j\in [n]} \Big( \vert y x_j\vert  < 1 \Big),
\end{split}
\]
where $C= D^{n^c}$ and $c$ is a fixed integer constant, independent of $f$.
\end{theorem}

\begin{proof}
Assume that $\Phi$ is true. Let $x_1,\dots,x_n\in\RN$ be values such that $f(x_1,\dots,x_n) = 0$. For sufficiently small $y'>0$
we have $\vert y' x_j\vert  < 1$ for all $j\in [n]$. For every $z>0$ we can then
take $y=\min\{z/2,y'\}$ to satisfy $\Psi$.

Assume that $\Psi$ is true and that $\Phi$ is false. Since polynomials are continuous we can assume that
$f(x_1,\dots,x_n)>0$ for all $x_1,\dots,x_n\in\RN$.  Let $y_1,y_2,\dots$ be a sequence of positive reals in $(0,1]$ such that $y_{k+1}<y_k/2$. Corresponding to these we choose
$x_1^{(k)},\dots,x_n^{(k)}$ that satisfies $\Psi$ for $y_k$, that is,
\begin{equation}\label{eq:strict1}
\Big(\vert f(x_1^{(k)},\dots,x_n^{(k)})\vert < y_k^C\Big)\wedge \bigwedge_{j\in [n]} \Big( \vert y_k x_j^{(k)}\vert  < 1 \Big).
\end{equation}
Then  $(x_1^{(k)},\dots,x_n^{(k)})_{k=1}^{\infty}$, as an infinite sequence in the closed space $[-\infty,\infty]^{n}$, contains a convergent subsequence; abusing notation we will use the same sequence for the convergent subsequence.
Let $(z_1,\dots,z_n)\in [-\infty,\infty]^n$ be the limit of the sequence.

We next apply a transformation that removes infinite limit points. To that end, let $S$ be the set of indices $i\in [n]$ for which $z_i$ is finite, $T=1+\max_{i\in S} \vert z_i\vert $, and $\Omega=\prod_{j\in [n]} [s_j,t_j]$
where $s_j=-T,t_j=T$ for $j\in S$, and $s_j=0,t_j=T$, for $j\not\in S$. Define
\[
\sigma_j(x) = \Bigg\{ \begin{array}{rl}
                      x & \mbox{if}\ z_j\in [-T,T], \\
                      1/x & \mbox{if}\ z_j=\infty, \\
                      -1/x & \mbox{if}\ z_j=-\infty,
                    \end{array}
\]
for $j\in [n]$,
\[g(x_1,\dots,x_n)=\prod_{j\in [n]\setminus S} x_j,\]
and
\[
W(x_1,\dots,x_n) = \left(g(x_1,\dots,x_n)\right)^D f(\sigma_1(x_1),\dots,\sigma_n(x_n)).
\]
Then $W(x_1,\dots,x_n)$ is a polynomial of total degree at most $D(n+1+$. Since $f$ is positive, it
follows that $W(x_1,\dots,x_n)=0$ implies $g(x_1,\dots,x_n)=0$ for all $(x_1,\dots,x_n)$.
By Theorem~\ref{thm:Linq}, there exists $\alpha>0$ such that
for all $(y_1,\dots,y_n)\in\Omega$ we have
\begin{equation}\label{eq:strict3}
W(y_1,\dots,y_n)\geq\alpha (g(y_1,\dots,y_n))^L,
\end{equation}
where $L=((n+1)D)^{n^{c''}} \leq D^{n^{c'}}$ for some $c'$ not depending on $f$.

Let $(\hat{z}_1,\dots,\hat{z}_n) = (\sigma_1(z_1),\dots,\sigma_n(z_n))\in\Omega$. Note that
$$\Big(\sigma_1(x_1^{(k)}),\dots,\sigma_n(x_n^{(k)})\Big)_{k=1}^{\infty}$$
converges to $(\hat{z}_1,\dots,\hat{z}_n)$ and for all sufficiently large $k$ we have
$$\big(\sigma_1(x_1^{(k)}),\dots,\sigma_n(x_n^{(k)})\big)\in\Omega.$$

Using $\vert x_j^{(k)}\vert < 1/y_k$ we obtain
\begin{equation}\label{eq:strict2}
g( \sigma_1(x_1^{(k)}),\dots, \sigma_n(x_n^{(k)})) \geq y_k^{n-\vert S\vert } \geq y_k^{n}.
\end{equation}
We have
\begin{equation}\label{eq:strict4}
\begin{split}
W(\sigma_1(x_1^{(k)}),\dots,\sigma_n(x_n^{(k)}))\left(g(\sigma(x_1^{(k)}),\dots,\sigma(x_n^{(k)}))\right)^{-D} = \\
 f(x_1^{(k)},\dots,x_n^{(k)}) < y_k^C.
\end{split}
\end{equation}
Using~\eqref{eq:strict3}, with $y_i=\sigma_i(x_i^{(k)})$, in~\eqref{eq:strict4} we obtain
$$
\alpha \left(g(\sigma(x_1^{(k)}),\dots,\sigma(x_n^{(k)}))\right)^{L-D} < y_k^C.
$$
Using~\eqref{eq:strict2} in the last equation we obtain
$$
\alpha y_k^{n(L-D)} < y_k^C.
$$
which is false for sufficiently large $y_k$, since $C-n(L-D)\geq 1$ for some $c$ not depending on $f$.
\end{proof}

We first use Theorem~\ref{thm:strictify2} to prove Theorem~\ref{thm:strictify}; we then show how to sharpen the result in Proposition~\ref{prop:strictcomp}.

\begin{proof}[Proof of Theorem~\ref{thm:strictify}.]
  Proposition~\ref{prop:QPOLY} implies that \SPOLY{k}{\#}\ and \PPOLY{k}{\#}\ are complete for $\SR{k}{}$ and
  $\PR{k}{}$. If the last quantifier block is universal, then the matrix is strict, so the result follows in this case.

  This leaves us with the case that the last quantifier block is existential. For $k=1$ the result $\SR{1}{} = \SR{1}{<}$ was proved in~\cite{SS17}, hence in all remaining cases there is at least one quantifier alternation, and the final block of quantifiers is existential.

  The complete problems \SPOLY{k}{\#}\ and \PPOLY{k}{\#}\ in these cases have a matrix of the form $g(x_1, \ldots, x_i) = 0$. Without loss of generality, let us assume that the variables existentially quantified in the final block are $x_1, \ldots, x_n$ for some $n < i$. Fix arbitrary values for the remaining variables $x_{n+1}, \ldots, x_i \in \RN$. Let $f(x_1, \ldots, x_n)$ be the polynomial resulting from fixing $x_{n+1}, \ldots, x_i$ in $g(x_1, \ldots, x_i)$. Theorem~\ref{thm:strictify2} implies that
\[(\exists x_1,\dots,x_n)\ f(x_1,\dots,x_n) = 0\]
is equivalent to the strict formula
\[
(\forall z>0) (\exists y\in (0,z)) (\exists x_1,\dots,x_n)\ \Big(\vert f(x_1,\dots,x_n)\vert < y^C\Big)\wedge \bigwedge_{j\in [n]} \Big( \vert y x_j\vert  < 1 \Big),\]
where $C=n D^{n^c} + (D+1) n $. But then
\[(\exists x_1,\dots,x_n)\ g(x_1,\dots,x_i) = 0\]
is equivalent to
\[
(\forall z>0) (\exists y\in (0,z)) (\exists x_1,\dots,x_n)\ \Big(\vert g(x_1,\dots,x_i)\vert < y^C\Big)\wedge \bigwedge_{j\in [n]} \Big( \vert y x_j\vert  < 1 \Big),\]
for the fixed values of $x_{n+1}, \ldots, x_i$. Since those values were chosen arbitrarily, and nothing in the two formulas depends on those values, we conclude that the two formulas are equivalent for all values of $x_{n+1}, \ldots, x_i \in \RN$. We can therefore replace the first formula with the second without affecting the truth of the full sentence. Since the final two quantifier blocks of our formula were of the type $\forall\exists$ this does not change the type of the formula, but the resulting formula is now strict.
\end{proof}

The complete problems produced by the previous proof contain Boolean matrices which, with some extra work, can be transformed into polynomial inequalities. A simple inequality, based on repeated squaring, helps reduce the degree of the polynomial.

\begin{lemma}\label{lem:ydexp}
The inequality
\begin{equation*}%\label{eq:ydexp}
\sum_{k\in[m]} (y_k - y_{k-1}^2-1)^2 < 1
\end{equation*}
is feasible (has a solution), and in any solution, $y_m > y_0^{2^m}$.
\end{lemma}
\begin{proof}
The inequality implies that $y_i > y_{i-1}^2$, from which the claim follows inductively.
\end{proof}

We can now establish the hardness of \SPOLY{k}{<}\ and \PPOLY{k}{<}.

\begin{proposition}\label{prop:strictcomp}
    \SPOLY{k}{<}\ is \SR{k}{}-complete, and \PPOLY{k}{<}\ is \PR{k}{}-complete for all $k \geq 1$. We can assume that the degree of the polynomial in the matrix is at most $6$ (at most $4$ if the final quantifier block is universal).
\end{proposition}
\begin{proof}
This follows immediately from Proposition~\ref{prop:QPOLY} if the last quantifier block is universal, so we can assume that the last quantifier block is existential. We will treat the case $k = 1$ separately in Lemma~\ref{lem:SPOLYk1}.

For $k \geq 2$, we follow the proof of Theorem~\ref{thm:strictify}. We start with a matrix of the form $g(x_1, \ldots, x_i) = 0$, where $g$ is a non-negative polynomial of degree at most $4$, and argue that
\[(\exists x_1,\dots,x_n)\ g(x_1,\dots,x_i) = 0\]
is equivalent to
\begin{equation}\label{eq:withH}
\begin{split}
(\forall z>0) & (\exists y\in (0,z)) (\exists x_1,\dots,x_n)\ \\
             &  \Big(\vert g(x_1,\dots,x_i)\vert < y^C\Big)\wedge \bigwedge_{j\in [n]} \Big( \vert y x_j\vert  < 1 \Big),
\end{split}
\end{equation}
for all $x_{n+1}, \ldots, x_i$. We need to reduce the unbounded degree term, $y^C$, and combine the multiple conditions into a single inequality. Let $m = \lceil \log_2(C)\rceil + 1$, so that $2^m > C$. Then Equation~\eqref{eq:withH} is equivalent to
\begin{equation}\label{eq:withoutH}
\begin{split}
(\forall z>0) & (\exists y_0\in (0,z)) (\exists y_1, \ldots, y_m) (\exists x_1,\dots,x_n)\ \\
             &  \sum_{k\in[m]} (y_k - y_{k-1}^2-1)^2 + y_m^2 g(x_1,\dots,x_i) + \sum_{j \in [n]} x_j^2y_m^2 < 1 .
\end{split}
\end{equation}
The matrix has degree at most $6$. As in Theorem~\ref{thm:strictify}, we can then replace $(\exists x_1,\dots,x_n)\ g(x_1,\dots,x_i) = 0$ with the strict condition in Equation~\eqref{eq:withoutH}.
\end{proof}

Corollary~\ref{cor:leq} follows: For a given complexity class $\CC$ we write $\co\CC$ for the class
of all problems whose complement belongs to $\CC$.
By definition, $\co\SR{k}{} = \PR{k}{}$, and vice versa, for all $k \geq 1$.

\begin{proof}[Proof of Corollary~\ref{cor:leq}.]
 We have $\co\PR{k}{<} \subseteq \SR{k}{\leq}$ by the definition of complementation. Since $\PR{k}{<} = \PR{k}{}$, by
 Theorem~\ref{thm:strictify}, we know that $\co\PR{k}{<} = \co\PR{k}{}=\SR{k}{}$, and $\SR{k}{} \subseteq \SR{k}{\leq}$ which implies $\SR{k}{\leq} = \SR{k}{}$. An analogous argument works for $\PR{k}{ \leq}$.
\end{proof}

Jungeblut, Kleist and Miltzow~\cite{JKM22} observed that
$\co\forall\exists_< \RN = \exists \forall_{\leq} \RN$; the corollary allows us to conclude that
both classes equal \SR{2}{}.

\medskip

To complete the proof of Proposition~\ref{prop:strictcomp} we still need to treat the \SR{1}{}-case. Ouaknine and Worrell~\cite[Theorem 7]{OW17}
showed that the bounded variant of \SPOLY{k}{<}, which we will treat in Section~\ref{sec:BU}, is \SR{k}{}-complete for $k = 1$, and write
that it is ``straightforward that the reduction [\,\dots] can be adapted [to the unbounded case]''; we have found that to be a bit tricky, so we decided
to include a proof based on the ideas in the current paper. For this proof, we need a version of Lemma~\ref{lem:ydexp} that
forces exponentially small rather than exponentially large values. At first glance it appears hard to achieve that using just
strict inequalities, but Lemma~\ref{lem:expsmall} does exactly that.
This lemma will also play a central role in Section~\ref{sec:BU} on bounded quantification.

\begin{lemma}\label{lem:expsmall}
Consider the inequality
\begin{equation}\label{eq:ineq}
400 \sum_{k \in [m]} (y_k - y_{k-1}^2)^2 < y_m^2,
\end{equation}
where $y_0 \in (0,1/2)$. The inequality~\eqref{eq:ineq} is feasible
with $y_i \in [-1,1]$ for all $i \in [m]$. For any solution $y_1,\dots,y_m$ of inequality~\eqref{eq:ineq} with $y_m \in [-1,1]$ we have
\begin{equation}\label{eq:yup}
0 < y_m^2 \leq y_0^{(5/6) 2^m}.
\end{equation}
\end{lemma}

\begin{proof}
To show feasibility take $y_k= y_0^{2^k}$. Note that this makes each term on the left-hand side of~\eqref{eq:ineq} equal to zero
and hence~\eqref{eq:ineq} is satisfied.

To argue~\eqref{eq:yup} suppose we have a solution $y_1,\dots,y_m$ of~\eqref{eq:ineq}.
From~\eqref{eq:ineq} it follows that for every $k\in [m]$ we have
\[
y_{k-1}^2 - \frac{\vert y_m \vert}{20} < y_k < y_{k-1}^2 + \frac{\vert y_m \vert}{20}.
\]
Suppose there exists $k\in\{1,\dots,m\}$ such that $\vert y_k \vert < \vert y_m \vert/2$. Then
\[
y_{k+1} < y_k^2 + \frac{\vert y_m \vert}{20} < \frac{\vert y_m \vert}{2} \left(\frac{\vert y_m \vert}{2} + \frac{1}{10}\right) \leq \frac{\vert y_m \vert}{2},
\]
(in the last inequality we used $\vert y_m \vert \leq 1$), and, similarly,
\[
-y_{k+1} < - y_k^2 + \frac{\vert y_m \vert}{20} \leq \frac{\vert y_m \vert}{2},
\]
yielding $\vert y_{k+1} \vert < \frac{\vert y_m \vert}{2}$. By induction, we conclude that $\vert y_m \vert<\vert y_m \vert/2$, a contradiction . Thus
$\vert y_k \vert\geq \vert y_m \vert/2$ for all $k\in [m]$. Hence,
\begin{equation}\label{ekk1}
y_{k-1}^2  < y_k + \frac{\vert y_m \vert}{20} \leq y_k +  \frac{\vert y_k \vert}{10},
\end{equation}
implying $y_k > 0$ for all $k \in [m]$. This allows us to drop the absolute values from this point on.
We also have
\begin{equation}\label{ekk2}
y_k\left(1-\frac{1}{10}\right) \leq y_k - \frac{y_m}{20} < y_{k-1}^2,
\end{equation}
for all $k\in [m]$. Define $x_k= \log_{y_0} y_k$. From ~\eqref{ekk1} and~\eqref{ekk2} we obtain that for
$k\in [m]$ we have
\[
\vert x_k - 2 x_{k-1}\vert  \leq 1/6;
\]
using $\log_{y_0}(9/10) \leq -\log_2(9/10)\leq 1/6$, $\log_{y_0}(11/10) \geq -\log_2(11/10)\geq -1/6$;
note that $x_0=1$. By induction we now have for $k\in [m]$
\begin{equation*}%\label{ekk2}
x_k \in [  (5/6) y_0^k,  (7/6) y_0^k ].
\end{equation*}
which implies~\eqref{eq:yup}.
\end{proof}

We can now treat the \SR{1}{}-case of Proposition~\ref{prop:strictcomp}.

\begin{lemma}\label{lem:SPOLYk1}
  \SPOLY{1}{<}\ is \SR{1}{}-complete. We can assume that the degree of the polynomial in the matrix is at most $6$.
\end{lemma}

\begin{proof}
 It is known that testing whether a non-negative polynomial $f: \RN^n \rightarrow \RN$ of degree at most $4$ has a root in $[-1,1]^n$ is \SR{1}{}-complete; moreover, one can guarantee that if $f$ has a root in $\RN^n$, it also has a root in $[-1,1]^n$. (The existence of such an $f$ follows, for example, from the construction in~\cite[Lemma 3.9]{S13} by summing up the squares of the quadratic polynomials $f_i$ created in that result.)

 If $f(x) > 0$ for all $x \in [-1,1]^n$, then, by the properties of $f$, we also have $f(x) > 0$ for all $x \in [-2,2]^n$.
 We can efficiently, in the length of the representation of $f$, find an integer $m$ such that if $f(x) > 0$ for all $x\in [-2,2]^n$, then
 $f(x) > 2^{-2^{m}}$ for all $x \in [-2,2]^n$, see, for example,~\cite[Corollary 3.7]{SS17}.
 %Note that $2^{-2^{m}} > 4^{-(5/6) 2^{m}}$, and consider the inequality
 Consider the inequality
 \begin{equation}\label{eq:fandx}
 \begin{split}
 400 &\left(\left(y_1-\left(\frac{1}{4}\right)^2\right)^2 +  \sum_{k \in [m]\setminus[1]} (y_k-y_{k-1}^2)^2\right)  \\
                             &+ f(x) + \sum_{k \in [n]} \left(x_k(1+z_k^2)-2z_k\right)^2 < y^2_m(1-y^2_m).\\
 \end{split}
\end{equation}
Suppose the inequality holds for some $(x,y,z) \in \RN^{n+m+n}$. The three summands on the left-hand side are all non-negative,
allowing us to drop the terms $f(x)$ and $\sum_{k \in [n]} (x_k(1+z_k^2)-2z_k)^2$, without affecting the truth of the inequality, to conclude
that $y^2_m < 1$, since otherwise, the right-hand side is negative. But then Lemma~\ref{lem:expsmall} applies
with $y_0 = \frac{1}{4}$ and using $y^2_m(1-y^2_m)< y^2_m$, so we can derive the upper bound $y^2_m \leq  4^{-(5/6) 2^{m}} < 2^{-2^{m}}$.

Inequality~\eqref{eq:fandx} implies that $(x_k(1+z_k^2)-2z_k)^2 < y^2_m$, so $\left(x_k-\frac{2z_k}{1+z^2_k}\right)^2 < y^2_m$; since the range of $z \mapsto \frac{2z}{1+z^2}$ is $[-1,1]$ for $z \in \RN$, and $y^2_m < 1$, we can conclude that $x_k \in [-2,2]$ for every $k \in [n]$.
Since Inequality~\eqref{eq:fandx} also implies that $f(x) < y^2_m \leq 2^{2^{-m}}$, by definition of $m$ we must have that $f(x) = 0$. By the properties of $f$, the existence of an $x \in [-2,2]^n$ with $f(x) = 0$ implies the existence of an $x \in [-1,1]^n$ with $f(x) = 0$.

On the other hand, if $f(x) = 0$ for some $x \in [-1,1]^n$, then Inequality~\eqref{eq:fandx} is feasible: for that $x$, $y_k = 4^{-2^{k}}$
and $z_k = \frac{1-(1-x_k)^{1/2}}{x_k}$, the left-hand side of Inequality~\eqref{eq:fandx} is zero, while the right-hand-side is positive.

In summary, $f$ has a root in $[-1,1]^n$ if and only if Inequality~\eqref{eq:fandx} holds, and the inequality can easily be transformed into the form $g(x) > 0$ for a
polynomial $g$ of degree $6$.
\end{proof}

\begin{remark}\label{rem:SPOLYk1deg4}
  The degree of the polynomials in Lemma~\ref{lem:SPOLYk1} can be lowered to $4$ by replacing the term $\sum_{k \in [n]} \left(x_k(1+z_k^2)-2z_k\right)^2$
  with
  \[\sum_{k \in [n]} \left(x_kz'_k-2z_k\right)^2 +  \left(z'_k-(1+z_k^2)\right)^2, \]
  and adding variables $z'_k$, $k \in [n]$. This still implies that $x_k \in [-2,2]$ and reduces the total degree.
\end{remark}

\subsection{Bounding the Universe}\label{sec:BU}

We say a $\Sigma_k$- or $\Pi_k$-formula is {\em bounded}, and write $\bd\Sigma_{k}$ and $\bd\Pi_{k}$ if the quantifiers range over a closed bounded set; there are several common choices for such a set, e.g.\ the unit simplex; we will instead restrict each variable in a bounded sentence to the interval $[-1,1]$. We write $\bd\SR{k}{}$ and $\bd\PR{k}{}$ for the bounded versions of \SR{k}{}\ and \PR{k}{}, and \bd\SPOLY{k}{}\ and \bd\PPOLY{k}{}\ for the bounded forms of \SPOLY{k}{}\ and \PPOLY{k}{}. We also consider {\em bounded open} variants in which quantification is over the open interval $(-1,1)$. For the bounded open variants we use the prefix ``\bowod'' in place of ``\bdwod''.

It follows from the arguments in~\cite{SS17} that $\SR{1}{} = \bd\SR{1}{}$, that is, $\NPR = \bd\exists\RN$. Ouaknine and Worrell~\cite[Theorem 7]{OW17} showed that $\SR{1}{} = \bd\SR{1}{<}$, D'Costa, Lefaucheux, Neumann, Ouaknine, and Worrell~\cite[Theorem 3]{DCLNOW21} showed that $\SR{2}{\leq} = \bd\SR{2}{\leq}$, and work by Jungeblut, Kleist and Miltzow~\cite[Theorem 12]{JKM22} implies
$\PR{2}{<} \subseteq \bd\PR{2}{}$.

%It turns out the open case is easier to handle than the closed case, for which we are left with open questions.
We handle the open case first in Section~\ref{sec:bouni}; Proposition~\ref{prop:boundedopen} in that section immediately implies
Corollary~\ref{cor:bound}. The closed case is treated in Section~\ref{sec:bduni}.
%, which deals with the bounded open case, and Proposition~\ref{prop:boundedclose}, which handles the closed case.

\subsubsection{The Bounded Open Universe}\label{sec:bouni}

We start by showing a version of Theorem~\ref{thm:strictify2} for quantification over bounded open domains.

\begin{theorem}\label{thm:strictify3}
Let $f(x_1,\dots,x_n)\in \RN[x_1,\dots,x_n]$ be of total degree at most $D \geq 2$, and $I = (-1,1)$.
Then the sentence
\[
\Phi = (\exists x_1,\dots,x_n \in I)\ f(x_1,\dots,x_n) = 0,
\]
is equivalent to
\[
\begin{split}
\Psi = (\forall z>0) & (\exists y\in (0,z)) (\exists x_1,\dots,x_n \in I)\ \\
                     & \Big(\vert f(x_1,\dots,x_n)\vert < y^C\Big)\wedge \bigwedge_{j\in [n]} ( y < 1-x^2_j ),
\end{split}
\]
where $C= D^{n^c}$ and $c$ is a fixed integer constant, independent of $f$.
\end{theorem}

\begin{proof}
If $\Phi$ is true, then there are $x_1,\dots,x_n\in I$ such that $f(x_1,\dots,x_n) = 0$. For sufficiently small $y'>0$
we have $y' < 1-x^2_j$ for all $j\in [n]$. For every $z>0$ we take $y=\min\{z/2,y'\}$ to satisfy $\Psi$.

Assume that $\Phi$ is false, so
$\lvert f(x_1,\dots,x_n)\rvert >0$ for all $x_1,\dots,x_n\in I$.
Let $g(x_1, \ldots, x_n) = \prod_{j \in [n]} (1-x_j^2)$. Then the conditions of Theorem~\ref{thm:Linq}, Solern\'{o}'s inequality, apply
since $g$ is zero on the boundary of $I^n$, and we can conclude that
\begin{equation}\label{eq:solstrict3}
\lvert f(x_1, \ldots,x_n)\rvert \geq \alpha \lvert g(x_1, \ldots, x_n) \rvert ^L,
\end{equation}
for all $x_1, \ldots, x_n \in I$, with $L = D^{n^{c'}}$, some fixed integer constant $c'$, independent of $f$,
and a constant $\alpha > 0$ which may depend on $f$ and $g$.

We want to show that $\Psi$ is false. To that end, choose $z \in (0,\alpha)$. Then $y \in (0,\alpha)$.
If $\bigwedge_{j\in [n]} ( y <(1-x_j^2))$ holds for some $x_1, \ldots, x_n \in I$,
then $g(x_1, \ldots, x_n) > y^n$. By~\eqref{eq:solstrict3}, we then have
\[f(x_1, \ldots, x_n) \geq \alpha g(x_1, \ldots,x_n)^L \geq \alpha y^{nL} > y^{nL+1},\]
which contradicts the first clause of $\Psi$, showing that $\Psi$ is false for $C = D^{n^{c}}$ for some $c$ independent of $f$.
\end{proof}

To obtain a degree-bounded version of Theorem~\ref{thm:strictify3} we want to replace the term $y^C$ which occurs in the bound with a fixed-degree computation. In the unbounded case we used Lemma~\ref{lem:ydexp} to achieve that goal, but that result required variables with unbounded values, which are not allowed in bounded quantification. Instead we need to compute doubly-exponentially small values, for which we will use Lemma~\ref{lem:expsmall}.

%We can now establish the bounded equivalent of Theorem~\ref{thm:strictify2}.

\begin{corollary}\label{cor:strictify2}
Let $f \in \RN[x_1,\dots,x_n]$ be a polynomial of total degree $d$, and let
$I = (-1,1)$.% or $I = [-1,1]$.
Then we can efficiently construct a polynomial $g \in \RN[x_1, \ldots, x_m]$ of degree
at most $\max(4,2d)$ (at most $\max(4,d)$ if $f$ is non-negative) such that
\[
\Phi = (\exists x_1,\dots,x_n \in I)\ f(x_1,\dots,x_n) = 0,
\]
is equivalent to
\[
\Psi = (\forall y\in I) (\exists x_1,\dots,x_m \in I)\ g(x_1, \ldots, x_m, y) > 0.
\]
If $f \in \ZN[x_1,\dots,x_n]$ then we can efficiently find a polynomial $g \in \RN(x_1, \ldots, x_m)$ of degree
at most $\max(4,2d)$ (at most $\max(4,d)$ if $f$ is non-negative)  such that $\Phi$ is equivalent to
\[
\Psi = (\exists x_1,\dots,x_m \in I)\ g(x_1, \ldots, x_m) > 0.
\]
In both cases, the strict condition $g>0$ in $\Psi$ can be replaced with $g \geq 0$.
\end{corollary}
\begin{proof}
If $f$ is {\em not} non-negative, we replace it by $f^2$. So by doubling the degree we can assume that $f$ is non-negative. Consider
\[
f^{*}(x_0,x_1, \ldots,x_n,x'_1,\ldots,x'_n) := f(x_1, \ldots, x_n) + \sum_{i \in [n]} (x_i^2 + {x'_i}^2 - x_{0})^2.
\]
If $f(x_1, \ldots, x_n) = 0$ for some $x_1, \ldots, x_n \in I$, then there are $x_0, x'_1, \ldots, x'_{n} \in I$ for which
$f^{*}(x_0, \ldots, x_n, x'_1, \ldots, x'_n) = 0$: choose $x_0 = \max\{x^2_i: i \in [n]\}$ and let
$x_i' = (x_0-x_i^2)$; moreover, every solution of $f^*(x_0, x_1,\ldots, x_n, x'_1, \ldots, x'_n) = 0$ in $I^{2n+1}$
satisfies $x_0 \geq \max\{x^2_i, {x'_i}^2: i \in [n]\}$.
In other words, by replacing $f$ with $f^*$, we can assume that the first argument of $f$ is an upper bound on the squares of all other arguments for any solution of $f = 0$.

By Theorem~\ref{thm:strictify3}, we have
\[
\Phi = (\exists x_1,\dots,x_n \in I)\ f(x_1,\dots,x_n) = 0,
\]
is equivalent to
\begin{equation*}%\label{eq:bdequiv}
\begin{split}
(\forall z>0) & (\exists y\in (0,z)) (\exists x_1,\dots,x_n \in I)\ \\
                     & \Big(\vert f(x_1,\dots,x_n)\vert < y^C\Big)\wedge \bigwedge_{j\in [n]} ( y < 1-x^2_j ),
\end{split}
\end{equation*}
where $C= D^{n^c}$. Using that $x_1 \geq \max\{x^2_i, {x'_i}^2: i \in [n]\}$ in a solution of
$f(x_1, \ldots, x_n) = 0$ we can conclude that $(1-x_1)/2 < 1-x_1^2$ as well as
$(1-x_1)/2  < 1-x_1 \leq 1-x_i^2$. This allows us to replace $\bigwedge_{j\in [n]} ( y < 1-x^2_j )$ with the stricter, but still satisfiable, condition $y < \frac{1-x_1}{2}$, which
gives us a simplified, equivalent formulation:
\begin{equation*}%\label{eq:bdequiv}
\begin{split}
(\forall z>0) & (\exists y\in (0,z)) (\exists x_1,\dots,x_n \in I)\ \\
                     & \Big(\vert f(x_1,\dots,x_n)\vert < y^C\Big)\wedge  y < \frac{1-x_1}{2}.
\end{split}
\end{equation*}

We need to replace the quantification domains of $z$ and $y$ with $I$.
To do so, we replace $z$ with $1-z$ and $y$ with $1-y$.% and change the range of $y$ from $(0,z)$ to $(0,z/2)$---this does not affect the truth of the statement.
\begin{equation}\label{eq:bdequiv}
\begin{split}
(\forall z \in I) & (\exists y \in I) (\exists x_1,\dots,x_n \in I)\ \\
                     & (1-y) < (1-z) \wedge \Big(\vert f(x_1,\dots,x_n)\vert < (1-y)^C\Big)\wedge  1-y < \frac{1-x_1}{2} .
\end{split}
\end{equation}

Pick the smallest $m \in \NN$ for which $(5/6)2^m > C$. Consider the following sentence:
\begin{equation}\label{eq:bdequiv2}
\begin{split}
(\forall z\in I) (\exists y_0,y_1, \ldots,  & y_m \in I) (\exists x_1,\dots,x_n \in I)   \\
        2(1-y_0) y_m^2  +  & 400\sum_{k \in [m]} (y_k - y_{k-1}^2)^2 +  f(x_1, \ldots, x_n) < \\
        & y_m^2 \frac{(1-z)(1-x_1)}{4}. \\
\end{split}
\end{equation}
If $\Phi$ is true, we can pick $x_1, \ldots, x_n \in I$ for which $f(x_1, \ldots, x_n) = 0$; choose $y_0 \in I$ such that $2(1-y_0) < (1-z)(1-x_1)/4$, and let $y_k = y_{k-1}^2$; these values show that~\eqref{eq:bdequiv2} is true. On the other hand, if~\eqref{eq:bdequiv2} is true, then in particular, $2(1-y_0)y_m^2 < y_m^2(1-z)(1-x_1)/4$ which implies both $1-y_0 < 1-z$ and $1-y_0 < (1-x_1)/2$, verifying
two of the conditions in~\eqref{eq:bdequiv}. We also get $2(1-y_0) < 1$, which implies that $y_0 < 1/2$, so we can use Lemma~\ref{lem:expsmall}
to conclude that $f(x_1,\ldots,x_n)<(1-y)^C$ which was the last condition in~\eqref{eq:bdequiv}
we had to verify. It follows that $\Phi$ is equivalent to~\eqref{eq:bdequiv2}.

We can then define $g$ as $y_m^2(1-z)(1-x_1) - 4\bigl(2(1-y_0^2)y_m^2+ 400 \sum_{k \in [m]} (y_k - y_{k-1}^2)^2 + f(x_1, \ldots, x_n)\bigr)$; rename the variables of $g$ as $x_1, \ldots, x_m$, with a new $m$ to get $g$ as required in the statement of the theorem; $g$ has degree $\max(4,d)$ or $\max(4,2d)$ depending on whether the original $f$ was non-negative.

Consider the case $f \in \ZN[x_1,\dots,x_n]$; as above we can assume that $f(x_1, \ldots, x_n) = 0$ for $x_1, \ldots, x_n \in (-1,1)^n$
implies that $x_1 \geq \max\{x^2_i, {x'_i}^2: i \in [n]\}$. By continuity, the same conclusion still holds in $[-1,1]^n$.

If $f(x_1, \ldots,x_n) \neq 0$ for all $x_1, \ldots, x_n \in (-1,1)$, then  $f(x_1, \ldots,x_n) = 0$ for $x_1, \ldots, x_n \in [-1,1]$
implies that $x^2_i = 1$ for some $i$, and so $x_1 = 1$. Hence, if we let $g = 1-x_1$ we can
apply Solern\'{o}'s inequality, Theorem~\ref{thm:Linq} to $f$ and $g$ on $[-1,1]^n$.
If $f(x_1,\ldots,x_n) \neq 0$ for all $x_1, \ldots, x_n \in (-1,1)$, the inequality tells us that
\begin{equation*}
\begin{split}
\vert f(x_1,\dots,x_n)\vert &\geq \alpha \vert g(x_1,\dots,x_n)\vert ^{L} = \alpha (1-x_1)^L = 2^{-\ell D^{cn^2}} (1-x_1)^{D^{n^c}}\\
        &\geq \left(\frac{1-x_1}{2}\right)^{\ell D^{cn^{c+2}}} ,
\end{split}
\end{equation*}
for all $x_1, \ldots, x_n \in [-1,1]$, where $\ell$ is the number of bits in the longest coefficient of $f$. We pick the smallest $m \in \NN$ so
that $(5/6)2^m > \ell D^{cn^{c+2}}$. Then $\Phi$ is equivalent to
\begin{equation}\label{eq:bdequiv3}
\begin{split}
(\exists y_0 \in I) & (\exists y_1, \ldots, y_m \in I) (\exists x_1,\dots,x_n \in I)   \\
        &  y_0^2 y_m^2 + 400 \sum_{k \in [m]} (y_k - y_{k-1}^2)^2 + f(x_1, \ldots, x_n) < y_m^2\frac{1-x_1}{2}, \\
\end{split}
\end{equation}
since $y_0^2 y_m^2  <  y_m^2\frac{1-x_1}{2}$ implies that $y_0^2 < (1-x_1)/2 < 1/2$. Define $g$ as
$y_m^2(1-x_1) - \bigl(2y_0^2y_m^2+ 400 \sum_{k \in [m]} (y_k - y_{k-1}^2)^2 + (f(x_1, \ldots, x_n))^2\bigr)$ and rename the variables as in the earlier case; the same degree bounds apply.

Neither case depends on $g$ being strictly greater than $0$, the bounds still work as long as $g \geq 0$, so we can replace strict inequality with weak inequality.
\end{proof}

Proposition~\ref{prop:boundedopen} implies Corollary~\ref{cor:bound}.
%; we will see the corresponding Proposition~\ref{prop:boundedclose} for the closed case in the next section.

\begin{proposition}\label{prop:boundedopen}
\bo\SPOLY{k}{<}\ is \SR{k}{}-complete and \bo\PPOLY{k}{<}\ is \PR{k}{}-complete for all $k \geq 1$.
\bo\SPOLY{k}{=}\ is \SR{k}{}-complete for all odd $k \geq 1$, and \bo\PPOLY{k}{=}\ is \PR{k}{}-complete for even $k \geq 1$.
We can assume that the degree of the polynomial in the matrix is at most $8$ (at most $4$ if the final quantifier block is existential).
\end{proposition}

\begin{proof}
    By Proposition~\ref{prop:QPOLY} we can assume that we are starting
    with a sentence of the form
    \[(Q_1 x_1)(Q_2 x_2) \cdots (Q_i x_i)\ \varphi(x_1, \ldots, x_i),\]
    where $\varphi(x_1, \ldots, x_i)$ is of the form  $f(x_1, \ldots, x_i) \relop 0$, $f$ is a non-negative polynomial of degree at most $4$, and $\relop \in \{=,>\}$ depends on which class we are working with. We distinguish two cases based on the final quantifier block.

    {\bfseries\noindent Case $Q_i = \exists$.} So the final quantifier block is existential.
    Proposition~\ref{prop:QPOLY} then implies that the formula's matrix $\varphi(x_1, \ldots, x_i)$ has the form $f(x_1, \ldots, x_i) = 0$.
    %    where $f$ is a polynomial of degree at most $4$.
    Using the bijection $z \rightarrow \frac{z}{(1-z^2)}$
    between $(-1,1)$ and $\RN$ we replace each $x_j$ in $f$ with the term $\frac{z_j}{(1-z_j^2)}$.
    To get rid of the rational terms, we multiply the whole expression by $\Pi_{\ell \in [i]} (1-z_{\ell}^2)^4$.
    After cancellation, this leaves us with a sum of the original monomials of $f$ each multiplied
    by a product of terms of the form $(1-z_{\ell}^2)^d$ for some $\ell \in [i]$ and $d \in [4]$.
    Using a Tseitin-style reduction, we can reduce the total degree of this expression to $4$:
    Each term of the sum takes on a value in the range $(-1,1)$, since it is the product
    of terms with this property, so by adding existentially quantified variables
    we can calculate the value of each term as a sequence of products $r = st$ which
    we can encode as $(r-st)^2=0$. The value of $f$ is then the sum of the variables
    representing the terms, so we can build a new function $h$ which is
    the sum of these variables plus all the terms of the form $(r-st)^2$ we needed.
    It follows that the original sentence
    \[(Q_1 x_1)(Q_2 x_2) \cdots (Q_i x_i)\ f(x_1, \ldots, x_i) = 0,\]
    is equivalent to
    \begin{equation}\label{eq:prepstrict}
    \begin{split}
        (Q_1 z_1 &\in (-1,1))(Q_2 z_2 \in (-1,1)) \cdots (Q_i z_i \in (-1,1)) \\
                & (\exists z_{i+1}, \ldots, z_m \in (-1,1))\  h(z_1, \ldots, z_i, z_{i+1}, \ldots, z_{m}) = 0.
    \end{split}
    \end{equation}
    This implies that \bo\SPOLY{k}{=} is \SR{k}{}-complete for odd $k \geq 1$ and \bo\PPOLY{k}{=} is \PR{k}{}-complete for even $k \geq 1$; moreover, the polynomial $h$ is non-negative and has total degree at most $4$.

    We then apply Corollary~\ref{cor:strictify2} to the function resulting from $h$ in~\eqref{eq:prepstrict} by restricting it to the final block of existentially quantified variables, $z_{i+1}, \ldots, z_m$. If those are all variables, so $k = 1$, then we apply the second, integer, version of
    Corollary~\ref{cor:strictify2} to get an existentially quantified inequality over a bounded open domain. If there is quantifier alternation, so $k>1$, we apply the first, real, version of Corollary~\ref{cor:strictify2}. As in Theorem~\ref{thm:strictify}, we can absorb the $\forall\exists$ quantifiers in the final two quantifier blocks.

    This shows that \bo\SPOLY{k}{<} is \SR{k}{}-complete for odd $k \geq 1$ and \bo\PPOLY{k}{<} is \PR{k}{}-complete for even $k \geq 1$.

    \smallskip
    {\bfseries\noindent Case $Q_i = \forall$.} The final quantifier block is universal. Let $\Theta$ denote the sentence in this case. It follows from the proof of Proposition~\ref{prop:QPOLY} that the matrix $\varphi(x_1, \ldots, x_i)$ of $\Theta$ has the form $f(x_1, \ldots, x_i) >0$, where $f$ is a non-negative polynomial of degree at most $4$. By negating $\Theta$ we can then proceed as in the first case, since we can negate $f(x_1, \ldots, x_i) >0$ as $f(x_1, \ldots, x_i) = 0$, because $f$ is non-negative.
    Following the first case up to and excluding the last step (the application of Corollary~\ref{cor:strictify2}), we build a sentence equivalent to $\neg \Theta$, with the same quantifier structure as $\neg \Theta$ and
    a matrix of the form $h(z_1, \ldots, z_i, z_{i+1}, \ldots, z_{m}) = 0$. Negating that sentence, and negating the matrix condition as $h^2(z_1, \ldots, z_i, z_{i+1}, \ldots, z_{m}) > 0$ then gives us a strict, bounded sentence with the same quantifier structure as the original sentence $\Theta$.
\end{proof}

\subsubsection{The Bounded Closed Universe}\label{sec:bduni}

As mentioned earlier, the only results on bounded closed universes are $\SR{1}{} = \bd\SR{1}{}$~\cite{SS17},
$\SR{1}{} = \bd\SR{1}{<}$~\cite[Theorem 7]{OW17},
$\SR{2}{\leq} = \bd\SR{2}{\leq}$~\cite{DCLNOW21} and $\PR{2}{<} \subseteq \bd\PR{2}{}$~\cite{JKM22}.

Our results from Section~\ref{sec:bouni} allow us to locate \bd\SR{k}{}\ and \bd\PR{k}{}.

\begin{proposition}\label{prop:bcunb}
 $\bd\SR{k}{} = \SR{k}{}$ and $\bd\PR{k}{} = \PR{k}{}$ for all $k \geq 1$.
\end{proposition}
\begin{proof}
Let $\varphi$ be a formula of type \bo\SPOLY{k}{<}, that is, $\varphi$ is of the form
\begin{equation}\label{eq:bcunb1}
(Q_1 x_1 \in (-1,1))\cdots(Q_i x_i \in  (-1,1))\ f(x_1, \ldots,x_i)>0
\end{equation}
with $k$ quantifier blocks and the first block existential.
By Proposition~\ref{prop:boundedopen}, deciding the truth of such a formula is \SR{k}{}-complete.
Let $U = \{j: Q_j = \forall\}$ and $E = \{j:Q_j = \exists\}$. We can then rewrite~\eqref{eq:bcunb1} equivalently as
\begin{equation}\label{eq:bcunb2}
\begin{split}
(Q_1 x_1 \in [-1,1])\cdots& (Q_i x_i \in  [-1,1])\ \\
& \bigvee_{j \in U} x^2_j = 1 \vee \bigwedge_{j \in E}
 x^2_j < 1 \wedge  f(x_1, \ldots,x_i)>0,
 \end{split}
\end{equation}
which is a formula in \bd\SR{k}{}, showing that $\SR{k}{} \subseteq \bd\SR{k}{}$. Since \bd\SR{k}{} is a special case of \SR{k}{}, it follows that
 $\bd\SR{k}{} = \SR{k}{}$. The second claim $\bd\PR{k}{} = \PR{k}{}$ follows similarly.
\end{proof}

As soon as we restrict the signature, to $<$ or $=$, say, the bounded-closed case becomes much harder to handle than
the bounded-open case. The methods developed for the first level do not seem to generalize easily to this setting, and new ideas seem required. So we cannot locate the bounded-closed classes like \bd\SR{k}{=}\
and \bd\SR{k}{<}, but we can show that the bounded-closed classes with restricted signatures are the same.

\begin{proposition}\label{prop:bceq}
 $\bd\SR{k}{=} = \bd\SR{k}{<}$ and \bd\SPOLY{k}{=}\ and \bd\SPOLY{k}{<}\ are complete for these classes for odd $k \geq 1$.
 $\bd\PR{k}{=} = \bd\PR{k}{<}$ and \bd\PPOLY{k}{=}\ and \bd\PPOLY{k}{<}\ are complete for these classes for even $k \geq 2$. All polynomials are of degree at most $4$.
\end{proposition}

The proof of the proposition requires several ingredients; the first is a restricted version of Tseitin's Lemma~\ref{lem:Tseitin} for the bounded closed case.

\begin{lemma}\label{lem:Tseitinbd}
   From a Boolean formula $\varphi(x_1, \ldots, x_i)$ in real variables $x_1, \ldots, x_i$ without negation and only allowing comparison operators in $\{=,\leq\}$ we can efficiently construct a
   family of quadratic polynomials $f_{\ell}:\RN^{i+j}\rightarrow \RN$, $\ell \in [k]$ such that
   \[\varphi(x_1, \ldots, x_i) \Leftrightarrow (\exists y_1, \ldots, y_j, \in [-1,1]) \bigwedge_{\ell \in [k]} f_{\ell} (x_1, \ldots, x_i, y_1, \ldots, y_j) = 0\]
  for all $x_1, \ldots, x_i \in [-1,1]$; each $f_{\ell}$ depends on at most four variables. %The parameters $j$ and $k$ are polynomial in the length of $\varphi$.
  Defining $f = \sum_{\ell \in [k]}f^2_{\ell}$, we obtain a non-negative polynomial $f: \RN^{i+j}\rightarrow \RN$ of degree at most $4$
  such that
  \[\varphi(x_1, \ldots, x_i) \Leftrightarrow (\exists y_1, \ldots, y_j \in [-1,1])\ f(x_1, \ldots, x_i, y_1, \ldots, y_j) = 0\]
  for all $x_1, \ldots, x_i \in [-1,1]$.
\end{lemma}

Our proof is similar to the proof of~\cite[Lemma 3.2]{SS17} which is based on Tseitin's original work~\cite{T83}, but it require some additional considerations to achieve the bounded domains.

\begin{proof}
 %We can assume that $\varphi$ contains no negation operators, since we can push those to the atomic level and absorb them in the atomic conditions, e.g.\ $\neg t > 0$ turns into $-t\geq 0$; we can restrict comparisons to $>$, $\geq$ and $=$.
 We can assume that the only integer constants in $\varphi$ are $0$ and $1$, see the earlier discussion on encodings. It is then easy to see that each term or subterm occurring in $\varphi$ can have value at most $2^{\ell}$, where $\ell$ is the bitlength of $\varphi$: this is true for the atomic terms $0$, $1$, and $x_i$, and remains true inductively, since there are fewer than $\ell$ arithmetical operations performed on these values.

 We start by creating new variables $z_1, \ldots, z_{\ell}$ and adding the constraints
 $z_1+z_1 = 1$, $z_i = z_{i-1}\cdot z_1$ for $i \in [\ell]$ so that $z_{\ell}$ evaluates to $2^{-\ell}$.

 For every subformula $\gamma$ of $\varphi$ we create a new variable $y_{\gamma}$ (and sometimes $y'_{\gamma}$); for every subterm $t$ occurring in an atomic formula we create a new variable $y_t$. We will create a collection of quadratic constraints in variables $y_{\gamma},y'_{\gamma}$ and $y_t$ such that if all constraints are satisfied, then $y_t = 2^{-\ell} \cdot t$ for all subterms $t$ and if $y_{\varphi} = 1$, then $\varphi$ is true.

 We add the following constraints. For every subformula $\gamma$ we add $y_{\gamma}y_{\gamma} - y_{\gamma} = 0$ (forcing $y_{\gamma} \in \{0,1\}$. If $\gamma$ is
 \begin{itemize}
  \item $\alpha\wedge \beta$, we add $y_{\gamma} = y_{\alpha}\cdot y_{\beta}$,
  \item $\alpha\vee \beta$, we add $y_{\gamma} = 1-(1-y_{\alpha})\cdot (1-y_{\beta})$.
 \end{itemize}
 For every subterm $t$ of the form
  \begin{itemize}
  \item $t = 0$, we add $y_t = 0$,
  \item $t = 1$, we add $y_t = z_{L}$,
  \item $t = x_i$, we add $y_t = x_i \cdot z_{L}$,
  \item $t = u+v$, we add $y_t = y_u+y_v$,
  \item $t = u-v$, we add $y_t = y_u-y_v$,
  \item $t = u\cdot v$, we add $y_t\cdot z_{L} = y_u \cdot y_v$,
 \end{itemize}
 We finally need to deal with atomic formulas, which are just term comparisons. If $\gamma$ is of the form
  \begin{itemize}
  \item $t = 0$, we add $y_{\gamma}\cdot y_t = 0$,
  \item $t \geq 0$, we add $y_{\gamma} \cdot (y_t- y'_{\gamma}) = 0$ and $y'_{\gamma} = y''_{\gamma}\cdot y''_{\gamma}$,
 \end{itemize}

 If $\varphi$ is true, we let $y_t = t\cdot z_{\ell}$ and $y_{\gamma} = 1$ if $\gamma$ is true and $0$ otherwise. We claim that this satisfies all constraints. This is immediate for the Boolean constraints. For the subterm constraints, note that $t\cdot z_{\ell} < 1$, by the upper bound on $t$ above, so $y_t \in [-1,1]$ for all terms $t \in T$.
 It remains to check the comparison constraints: for $t = 0$ we can let $y_{\gamma} = 1$, for $t \geq 0$ we can let $y'_{\gamma} = t$, $y''_{\gamma} = \sqrt{t}$ allowing $y_{\gamma} = 1$.

 For the other direction, suppose all constraints are satisfied. We claim that letting $x_i = y_{x_i}/z_{\ell}$ makes $\varphi$ true. First of all, we note that the subterm constraints ensure that $y_t = t \cdot z_{\ell}$ for all subterms $t$ of $\varphi$, and the Boolean constraints ensure that the logical connectives get evaluated correctly, so we are left with checking the atomic comparisons. Since $\varphi$ is monotone, it is sufficient to show that $y_{\gamma} = 1$ implies that $\gamma$ is true: it $y_{\gamma} = 0$, it cannot make a formula containing $\gamma$ as a subformula true, so it does not matter whether $\gamma$ is actually true.

 Suppose $\gamma$ is the formula $t = 0$. If $y_{\gamma} = 1$, then  $y_{\gamma}\cdot y_t = 0$ implies
 that $y_t = 0$, so $y_{\gamma} = 1$ is justified.
 In the next case, $\gamma$ is $t \geq 0$ and we proceed similarly. If $y_{\gamma} = 1$, then $(y_t- y'_{\gamma}) = 0$,
 so $y_t = y'_{\gamma}$ and $y'_{\gamma} = y''_{\gamma}\cdot y''_{\gamma}$, so $y_t$, and therefore $t(x_1, \ldots, x_i)$, is at least $0$.
\end{proof}

At the heart of Proposition~\ref{prop:bceq} is the following translation lemma.

\begin{lemma}\label{lem:bcineqtoeq} The following are true.
\begin{itemize}
\itemi A formula of type \bd\SR{k}{<}, where $k$ is odd, can efficiently be transformed into an equivalent formula in \bd\SR{k}{=}. A formula of type \bd\PR{k}{<}, where $k$ is even, can efficiently be transformed into an equivalent formula in \bd\PR{k}{=}.
\itemii A formula of type \bd\SPOLY{k}{=}, where $k$ is odd, can efficiently be transformed into an equivalent formula of type \bd\SPOLY{k}{<}. A formula of type \bd\PPOLY{k}{=}, where $k$ is even, can efficiently be transformed into an equivalent formula of type \bd\PPOLY{k}{<}. In both cases it is true that if the original polynomial is non-negative, the degree does not increase; otherwise it doubles.
\end{itemize}
\end{lemma}

We can now complete the proof of the main result in this section.

\begin{proof}[Proof of Proposition~\ref{prop:bceq}.]
Part $(i)$ of Lemma~\ref{lem:bcineqtoeq} implies that $\bd\SR{k}{<} \subseteq \bd\SR{k}{=}$ for odd $k\geq 1$.
Lemma~\ref{lem:Tseitinbd} shows that \bd\SPOLY{k}{=}\
is complete for \bd\SR{k}{=}. So part $(ii)$ of Lemma~\ref{lem:bcineqtoeq} then implies that
formulas of type \bd\SPOLY{k}{<}\ are hard for \bd\SR{k}{=}, and that $\bd\SR{k}{=} \subseteq \bd\SR{k}{<}$.
We conclude that $\bd\SR{k}{=} = \bd\SR{k}{<}$ for odd $k \geq 1$.

For even $k \geq 2$ we similarly argue that
$\bd\PR{k}{<} \subseteq \bd\PR{k}{=}$, by Lemma~\ref{lem:bcineqtoeq}$(i)$, and that \bd\SPOLY{k}{=}\ is complete for \bd\PR{k}{=}, by Lemma~\ref{lem:Tseitinbd}, so \bd\PPOLY{k}{<}\ is hard for \bd\PR{k}{=}, by Lemma~\ref{lem:bcineqtoeq}$(ii)$, so
$\bd\PR{k}{=} = \bd\PR{k}{<}$.

In all cases the polynomials the original polynomials obtained from Lemma~\ref{lem:Tseitinbd} are non-negative of degree at most $4$, so that when we invoke Lemma~\ref{lem:bcineqtoeq}$(ii)$, the degree does not increase, though the resulting polynomial will no longer be non-negative.
\end{proof}

Before we can prove Lemma~\ref{lem:bcineqtoeq} we need an intermediate lemma that allows us to replace strict inequalities with equalities in the bounded closed case, without appealing to Solern\'{o}'s inequality (which does not seem to adapt to this setting).

\begin{lemma}\label{lem:Berge}
Given a formula $\varphi$ of the form
\begin{equation}\label{eq:Berge1}
(Q_1 x_1 \in [-1,1]^{n_1})\cdots(Q_k x_k \in  [-1,1]^{n_k})\ f(x_1, \ldots,x_k)>0,
\end{equation}
where $f$ is a continuous function; then there is an $\alpha \in \RN_{>0}$ such that~\eqref{eq:Berge1} is equivalent to
\begin{equation}\label{eq:Berge2}
(Q_1 x_1 \in [-1,1]^{n_1})\cdots(Q_k x_k \in  [-1,1]^{n_k})\ f(x_1, \ldots,x_k) \geq \alpha,
\end{equation}
and also equivalent to
\begin{equation}\label{eq:Berge2b}
(Q_1 x_1 \in [-1,1]^{n_1})\cdots(Q_k x_k \in  [-1,1]^{n_k})\ f(x_1, \ldots,x_k) > \alpha.
\end{equation}
If $f(x_1, \ldots, x_k) = y$ is equivalent to a quantifier-free Boolean formula $\psi(x_1,\ldots,x_k,y)$, then $\alpha$ can be chosen to be $2^{-\ell^{{(cn)^k}}}$, where $n = \sum_{i=1}^k n_i$, $\ell$ is the bitlength of $\psi$ and $c$ is a universal constant independent of $f$, $n$ and $k$.
\end{lemma}
\begin{proof}
Let $\varphi$ be a formula of the form~\eqref{eq:Berge1}.
If $Q_i$ is universal, let $\opt^i = \inf$ , and $\opt^i = \sup$ otherwise. With that notation we introduce
the family of functions
\[g_i(x_1, \ldots, x_{i-1}) := \sideset{}{^i}\opt_{x_i \in [-1,1]^{n_i}} \cdots \sideset{}{^k}\opt_{x_k \in [-1,1]^{n_k}} f(x_1, \ldots,x_k),\]
and observe that $\alpha' := g_0$ is a constant satisfying
\begin{equation}\label{eq:Berge3}
(Q_1 x_1 \in [-1,1]^{n_1})\cdots(Q_k x_k \in  [-1,1]^{n_k})\ f(x_1, \ldots,x_k) \geq \alpha'.
\end{equation}
Since $f$ is continuous and all domains are compact, we can apply Berge's maximum theorem~\cite[Section 6.3]{B63} to show that the functions $g_k,g_{k-1}, \ldots, g_0$ are all continuous; it follows that if we define $h_i$ like $g_i$ but with $\min$ instead of $\inf$ and $\max$ instead of $\sup$, we have $h_i = g_i$ for all $i \in [k]$. In particular, $\alpha' = h_0$ so the value of $\alpha'$ is achieved over the compact domains of the quantifiers, namely, we have
\begin{equation}\label{eq:Berge4}
(Q_1 x_1 \in [-1,1]^{n_1})\cdots(Q_k x_k \in  [-1,1]^{n_k})\ f(x_1, \ldots,x_k) = \alpha'.
\end{equation}
If $\varphi$ is true, then~\eqref{eq:Berge1} and~\eqref{eq:Berge4} imply that $\alpha' > 0$, so~\eqref{eq:Berge1} is equivalent to
\begin{equation}\label{eq:Berge5}
(Q_1 x_1 \in [-1,1]^{n_1})\cdots(Q_k x_k \in  [-1,1]^{n_k})\ f(x_1, \ldots,x_k) \geq \lvert\alpha'\rvert.
\end{equation}
and also equivalent to
\begin{equation}\label{eq:Berge5b}
(Q_1 x_1 \in [-1,1]^{n_1})\cdots(Q_k x_k \in  [-1,1]^{n_k})\ f(x_1, \ldots,x_k) > \lvert\alpha'\rvert/2,
\end{equation}
where $>$ can also be replaced with $\geq$.
We can then let $\alpha = \lvert\alpha'\rvert/2$ if $\varphi$ is true, and $1$ otherwise. This proves the first part of the statement of the lemma.

For the second part, define the set $A$ as
\begin{equation}\label{eq:Berge6}
 \{y \in \RN: (Q_1 x_1 \in [-1,1]^{n_1})\cdots(Q_k x_k \in  [-1,1]^{n_k})\ f(x_1, \ldots,x_k) \geq y\},
\end{equation}
then $\alpha' \in A$ by Equation~\eqref{eq:Berge4}. Recall that $f$ can be defined using a quantifier-free formula $\psi$ in the theory of the reals, so that~\eqref{eq:Berge6} can be expressed in the theory of the reals with $k$ quantifier blocks. Using a technique of B\"{u}rgisser and Cucker~\cite[Proof of Theorem 9.2]{BC09} or Jungeblut, Kleist, Miltzow~\cite[Lemma 3.4]{JKM22}, where we use that $f$ can be defined using a formula $\psi$, there is a strict lower bound $\alpha$ on the supremum of $A$, and therefore $\alpha'$, of the form $2^{-\ell^{{(cn)^k}}}$ for some universal constant $c > 0$, independent of $f$, $n$ and $k$, where $\ell$ is the  bitlength of $f$. We conclude that $\varphi$ is equivalent to
\begin{equation}\label{eq:Berge7}
(Q_1 x_1 \in [-1,1]^{n_1})\cdots(Q_k x_k \in  [-1,1]^{n_k})\ f(x_1, \ldots,x_k)-2^{-\ell^{{(cn)^k}}}\geq 0,
\end{equation}
which proves the second part of the lemma.
%where we could also replace the $\geq 0$ with $>0$ (since the lower bound was strict).
\end{proof}

\begin{proof}[Proof of Lemma~\ref{lem:bcineqtoeq}.]
We prove the two parts separately.

{\bfseries\noindent Part $(i)$.} Let $\varphi$ be a formula in \bd\SR{k}{<} where $k$ is odd, or in \bd\PR{k}{<}\ where $k$ is even. In both cases, the final quantifier block is existential. Since $\varphi$ does not contain negation, it consists of conjunctions and disjunctions only. Notice that if $g$ and $h$ are two functions, then $g>0 \wedge h>0$ is equivalent to $\min(g,h)>0$ and $g>0 \vee h>0$ is equivalent to $\max(g,h)>0$. We can therefore replace the matrix of $\varphi$ with a single condition of the form $f>0$, where $f$ is a continuous function (not necessarily a polynomial), and $f$ can be defined using a quantifier-free formula $\psi$. We can then apply Lemma~\ref{lem:Berge} to obtain that $\varphi$ is equivalent to
\begin{equation}\label{eq:bcineqtoeq1}
(Q_1 x_1 \in [-1,1]^{n_1})\cdots(Q_k x_k \in  [-1,1]^{n_k})\ f(x_1, \ldots,x_k) \geq 2^{-\ell^{{(cn)^k}}}.
\end{equation}
We want to replace the condition $f(x_1, \ldots,x_k)-2^{-\ell^{{(cn)^k}}}\geq 0$ with an equality, but since we only have variables in $[-1,1]$ we need to bound the value of the left-hand side of this inequality. We do so by scaling the value of $f(x_1, \ldots,x_k)$, that is, we replace the condition with
\[(\exists \lambda_f,\delta_f \in [-1,1]^2)\ \lambda_f^2f(x_1, \ldots,x_k)-2^{-\ell^{{(cn)^k}}}=\delta_f^2,\]
which is equivalent to $f(x_1, \ldots,x_k)-2^{-\ell^{{(cn)^k}}}\geq 0$. Hence, $\varphi$ is equivalent to
\begin{equation}\label{eq:bcineqtoeq2}
\begin{split}
(Q_1 x_1 \in [-1,1]^{n_1}) & \cdots(Q_k x_k \in  [-1,1]^{n_k}) \\
                           &(\exists\lambda_f,\delta_f \in [-1,1]^2)\ \lambda_f^2f(x_1, \ldots,x_k)-2^{-\ell^{{(cn)^k}}}=\delta_f^2.
\end{split}
\end{equation}
With additional existential variables we can replace the expression $2^{-\ell^{{(cn)^k}}}$ with a (computed) variable $\varepsilon$; this nearly places~\eqref{eq:bcineqtoeq2} into \bd\SR{k}{=}\ (for odd $k$) or \bd\PR{k}{=}\ (for even $k$), except that $f$ is a function defined using minima and maxima. However, we can now unpack the definition of $f$. E.g.\ if $f$ was $\min(g,h)$ we replace
\begin{equation}\label{eq:bcineqtoeq3}
 \lambda_f^2f(x_1, \ldots,x_k)-\varepsilon =\delta_f^2
\end{equation}
with
\[
\left(\lambda_g^2f(x_1, \ldots,x_k)-\varepsilon =\delta_g^2\right) \wedge \left(\lambda_h^2f(x_1, \ldots,x_k)-\varepsilon =\delta_h^2\right),\]
and if $f = \max(g,h)$ we replace~\eqref{eq:bcineqtoeq3} with
\[
\left(\lambda_g^2f(x_1, \ldots,x_k)-\varepsilon =\delta_g^2\right) \vee \left(\lambda_h^2f(x_1, \ldots,x_k)-\varepsilon =\delta_h^2\right)\]
 and so on recursively, adding the $\lambda$ and $\delta$ variables to the final block of existentially quantified variables. This turns the matrix of $\varphi$ back into its original form except that each strict inequality $f>0$ has been replaced with an equality of the form~\eqref{eq:bcineqtoeq3}.

\smallskip

{\bfseries\noindent Part $(ii)$.}
Let $\varphi$ be a formula in \bd\SPOLY{k}{=}\ for odd $k$ or in \bd\PPOLY{k}{=}\ for even $k$. So $\varphi$ is of the form
\begin{equation}\label{eq:bcineqtoeq6}
(Q'_1 x_1 \in [-1,1]^{n_1})\cdots(Q'_k x_k \in  [-1,1]^{n_k})\ f(x_1, \ldots,x_k)=0,
\end{equation}
where $Q'_k = \exists$. We can assume that $f$ is nonnegative (by replacing it with $f^2$ if necessary; this will lead to doubling the degree).
Negating~\eqref{eq:bcineqtoeq6}, we obtain that $\overline{\varphi}$ is of the form~\eqref{eq:Berge1} in Lemma~\ref{lem:Berge} with $Q_i = \neg Q'_i$ for $i \in [k]$. We can then apply the lemma to conclude that $\varphi$ is equivalent to
\begin{equation}\label{eq:bcineqtoeq7}
(Q'_1 x_1 \in [-1,1]^{n_1})\cdots(Q'_k x_k \in  [-1,1]^{n_k})\ f(x_1, \ldots,x_k)-2^{-\ell^{{(cn)^k}}} < 0,
\end{equation}
using the strict version~\eqref{eq:Berge2b} of the lemma. We want to turn the condition
\begin{equation}\label{eq:bcineqtoeq8a}
f(x_1, \ldots,x_k)-2^{-\ell^{{(cn)^k}}} < 0
\end{equation}
into an explicit polynomial condition, and we can do so by adding new variables. Choose $m$ such that
$(5/6)2^m > \ell^{(cn)^k}$. Consider the condition
\begin{equation}\label{eq:bcineqtoeq8b}
(\exists y_1, \ldots, y_{L} \in [-1,1])\ f(x_1, \ldots,x_k) + 400 \sum_{k \in [m]} (y_k - y_{k-1}^2)^2 < y_m^2.
\end{equation}
If $f(x_1,\ldots,x_k) = 0$, then, by Lemma~\ref{lem:expsmall},~\eqref{eq:bcineqtoeq8b} is true, since $400 \sum_{k \in [m]} (y_k - y_{k-1}^2)^2 < y_m^2$ is feasible. On the other hand, if~\eqref{eq:bcineqtoeq8b} is true, then, since $f$ is nonnegative, we have $400 \sum_{k \in [m]} (y_k - y_{k-1}^2)^2 < y_m^2$ and $0 < y_m^2 <  2^{-\ell^{{(cn)^k}}}$, by choice of $m$.
Then~\eqref{eq:bcineqtoeq8b} implies that $f(x_1, \ldots,x_k) < y_m^2 < 2^{-\ell^{{(cn)^k}}}$, so~\eqref{eq:bcineqtoeq8a} is true. We have shown that~\eqref{eq:bcineqtoeq8a} and~\eqref{eq:bcineqtoeq8b} are equivalent, so we can rewrite ~\eqref{eq:bcineqtoeq7} by replacing $f(x_1, \ldots,x_k)-2^{-\ell^{{(cn)^k}}} < 0$ with~\eqref{eq:bcineqtoeq8b}, merging the additional existential quantifier with $Q_k$. The resulting formula is of the form \bd\SPOLY{k}{<}\ if $k$ is odd, and
\bd\PPOLY{k}{<}\ otherwise.
\end{proof}

The final result in this section we already mentioned earlier: at the second level of the hierarchy, the bounded closed domain classes coincide with the unbounded classes as expected.

\begin{proposition}\label{prop:bc2nd}
\bd\SPOLY{2}{<}\ is \SR{2}{}-complete for polynomials of degree at most $8$,
\bd\SPOLY{2}{\leq}\ is \SR{2}{}-complete for polynomials of degree at most $4$,
and \bd\PPOLY{2}{<}\ and \bd\PPOLY{2}{=}\ are \PR{2}{}-complete for polynomials of degree at most $4$.
\end{proposition}

The proof of Proposition~\ref{prop:bc2nd} is an extension of a result by Jungeblut, Kleist and Miltzow~\cite{JKM24}; we will be using their results in our proof.

\begin{proof}
Let us first show that \bd\PPOLY{2}{<}\ is \PR{2}{}-complete.
By Proposition~\ref{prop:QPOLY} we can start with a sentence of the form
\begin{equation}\label{eq:bc2nd1}
(\forall x \in \RN^{n_1})(\exists y \in \RN^{n_2})\ f(x,y) > 0,
\end{equation}
where $f$ is a polynomial of degree at most $4$. By Lemma~3.6 from~\cite{JKM24} it follows that~\eqref{eq:bc2nd1} is equivalent to
\begin{equation}\label{eq:bc2nd2}
(\forall x \in [-C,C]^{n_1})(\exists y \in \RN^{n_2})\ f(x,y) > 0,
\end{equation}
where $C = 2^{2^{c_1 n^{k_1}}}$ for some integers $c_1, k_1,>0$ and $n = n_1+n_2$.
The equivalence of~\eqref{eq:bc2nd1} and~\eqref{eq:bc2nd2} is not affected if we replace $C$ with any larger value in~\eqref{eq:bc2nd2}. In particular,~\eqref{eq:bc2nd2} is equivalent to
\begin{equation}\label{eq:bc2nd3}
(\forall x \in [-C^3,C^3]^{n_1})(\exists y \in \RN^{n_2})\ f(x,y) > 0.
\end{equation}
We now apply Lemma 3.7 from~\cite{JKM24} to conclude that~\eqref{eq:bc2nd3} is equivalent to
\begin{equation}\label{eq:bc2nd4}
(\forall x \in [-C^3,C^3]^{n_1})(\exists y \in [-D,D]^{n_2})\ f(x,y) > 0,
\end{equation}
where $D = 2^{2^{c_2 n^{k_2}}}$ for some integers $c_2, k_2>0$. Let $C'$ and $D'$ be such that $C \leq C' \leq C^3$ and $D \leq D' \leq D^3$. We claim that~\eqref{eq:bc2nd1} is equivalent to
\begin{equation}\label{eq:bc2nd5}
(\forall x \in [-C',C']^{n_1})(\exists y \in [-D',D']^{n_2})\ f(x,y) > 0.
\end{equation}
For the forward implication we use that~\eqref{eq:bc2nd1} is equivalent to~\eqref{eq:bc2nd4}, which directly implies~\eqref{eq:bc2nd5}. On the other hand, if~\eqref{eq:bc2nd1} is false, then there is an $x \in [-C,C]^{n_1}$ such that $(\forall y \in \RN^{n_2})\ f(x,y) \leq 0$, so~\eqref{eq:bc2nd5} is false, and we have established equivalence
of~\eqref{eq:bc2nd1} and~\eqref{eq:bc2nd5}.

Consider the inequality
 \begin{equation}\label{eq:ymCandD}
 400 \left(\left(z_1-\left(\frac{1}{4}\right)^2\right)^2 +  \sum_{k \in [m]\setminus[1]} (z_k-z_{k-1}^2)^2\right) < z^2_m.\\
\end{equation}
We know that this inequality is feasible with $z_k = 2^{-2^{k+1}}$. Moreover, we can choose $m$ sufficiently large such that
$1/2 \leq 2^{2^{k+1}}z_k \leq 3/2$ for all $k \leq \max(c_1 n^{k_1}, c_2 n^{k_2})$. (Using
Lemma~\ref{lem:expsmall} one can show by induction that $\lvert z_k - 2^{-2^{k+1}}\rvert \leq (1/20) 2^{-(5/6)2^m} 4^k$ for
$k\leq m/4$; this implies the claim.) If we define $a_i = c_i n^{k_i}$, $i \in [2]$,
we then have that $C \leq z_{a_1}^{-1} \leq C^3$, since $z_{a_1}^{-1}$ will be close to $C^2$, and similarly, $D \leq z_{a_2}^{-1} \leq D^3$.

We can then rewrite~\eqref{eq:bc2nd5} equivalently as
\begin{equation}\label{eq:bc2nd6}
\begin{split}
(\forall x  \in [-1,1]^{n_1})& (\exists y \in [-1,1]^{n_2})(\exists z \in [-1,1]^m)\ \\
(z_{a_1}& z_{a_2})^4 f(x/z_{a_1}, y/z_{a_2}) > 0\ \wedge \\
    & 400 \left(\left(z_1-\left(\frac{1}{4}\right)^2\right)^2 +  \sum_{k \in [m]\setminus[1]} (z_k-z_{k-1}^2)^2\right) < z^2_m.
\end{split}
\end{equation}
The formula in~\eqref{eq:bc2nd6} belongs to \bd\PR{2}{<}, showing that this class is \PR{2}{}-complete. Since \bd\PPOLY{2}{<}\ is complete for \bd\PR{2}{<}\ by Proposition~\ref{prop:bceq}, it follows that \bd\PPOLY{2}{<}\ is \PR{2}{}-complete. Since $\bd\PR{2}{<} = \bd\PR{2}{=}$ it follows that \bd\PPOLY{2}{=}\ is also \PR{2}{}-complete. The complement of \bd\PR{2}{<}\ is \bd\SR{2}{\leq}, so \bd\SR{2}{\leq}\ is \SR{2}{}-complete.

The complement of \bd\PR{2}{=}, which is \SR{2}{}-complete, is contained in \bd\SR{2}{<}, so it follows that $\bd\SR{2}{<} = \SR{2}{}$. We can then take a \bd\PR{2}{=}-complete formula of the form $(\forall x \in [-1,1]^{n_1})(\exists y \in [-1,1]^{n_2})\ f(x,y) = 0$, and negate it as $(\exists x \in [-1,1]^{n_1})(\forall y \in [-1,1]^{n_2})\ f^2(x,y) > 0$ to see
that \bd\SPOLY{2}{<}\ is \SR{2}{}-complete with polynomials of degree at most $8$.
\end{proof}

\section{Properties of Semi-Algebraic Sets And Exotic Quantifiers}\label{sec:SAS}

Testing properties of semialgebraic sets, such as being convex or closed, very naturally leads to problems captured by the real hierarchy. The complexity of such problems was first studied in the real polynomial hierarchy based on the BSS-model, going back to the original papers introducing the BSS-model, see~\cite{BSS89,BCSS98}.

B\"{u}rgisser and Cucker~\cite[Section 9]{BC09} showed that many of these results, often with the same proofs, carry over to the real hierarchy, the discrete version of the BSS-hierarchy. We sharpen some of these results using the ideas from the previous section, and add some new results.

Various properties defy classification using standard quantifiers, so B\"{u}rgisser and Cucker made use of three ``exotic'' quantifiers, $H$, $\forall^*$ and $\exists^*$. We will define those quantifiers below, and show that in the real hierarchy the power of these quantifiers can sometimes be explained using standard quantifiers.

\subsection{First-Level Problems}

We start with a simple example: testing the convexity of a semialgebraic set. The complexity of this problem was first determined in a paper by Cucker and Rossell\'{o}~\cite{CR92}.

\begin{theorem}[Cucker, Rossell\'{o}~\cite{CR92}]\label{thm:convex}
  Testing the convexity of a (bounded) semialgebraic set is \coNPR-complete.
\end{theorem}
\begin{proof}
 A semialgebraic set $S = \{x: \varphi(x)\}$ is convex if
 \[(\forall x, y)(\forall t\in [0,1])\ \big(\varphi(x) \wedge \varphi(y) \implies \varphi(tx+(1-ty)\big),\]
 so testing convexity belongs to \coNPR.

 By Proposition~\ref{prop:boundedopen}, \bo\PPOLY{1}{>}\ is \coNPR-complete. So we can assume we are given a sentence  $\Phi = (\forall x \in (-1,1)^n)\ f(x) > 0$. We define two semialgebraic sets,
 \[T = \{x \in [-1,1]^n: x_i \in \{-1,1\}\ \mbox{for some $i \in [n]$}\},\]
 the bounding box of $[-1,1]^n$, and with that
 \[S = \{x \in (-1,1)^n: f(x) > 0\} \cup T.\]
 Then $\Phi$ is true if and only if $S$ is convex.
\end{proof}

Let's look at a more challenging problem: how hard is it to test whether a semialgebraic set is unbounded? The problem lies in \PR{2}{}, since $S$ is unbounded if for every $d > 0$ there is an $x \in S$ with $\lVert x\rVert > d$, but is it hard for \PR{2}{}?

The universal quantifier can be eliminated, because it only affects a single variable, and the quantified condition is monotone in that variable. This was proved by B{\"u}rgisser and Cucker~\cite[Theorem 9.2]{BC09}, and rediscovered in~\cite[Lemma 4.1]{SS17} (and used earlier in~\cite{S13}).\footnote{This elimination result relies on the discrete setting.} For B{\"u}rgisser and Cucker this is their first exotic quantifier, the infinitesimal quantifier $H$, where $(H \varepsilon)\ \varphi(\varepsilon,x)$ is defined to mean
\[(\exists \varepsilon' > 0)(\forall \varepsilon \in (0,\varepsilon'))\ \varphi(\varepsilon, x).\]
In plain English, $H \varepsilon$ can be read as ``for all sufficiently small values of $\varepsilon > 0$''. We treat $H$ as a quantifier over a single real number; all other quantifiers can range over tuples.

The unboundedness of a semialgebraic set $S = \{x: \varphi(x)\}$ is then equivalent to
\[(H \varepsilon) (\exists x)\ \varphi(x) \wedge \varepsilon \lVert x \rVert > 1.\]
Since the $H$ quantifier can be eliminated, the problem lies in \NPR.

\begin{theorem}[B{\"u}rgisser, Cucker~{\cite[Proposition 6.4,Corollary 9.4]{BC09}}]\label{thm:unbounded}
  Testing whether a semialgebraic set is unbounded is \NPR-complete.
\end{theorem}
\begin{proof}
 We already saw membership. For hardness, we can start with the sentence $\Phi = (\exists x \in (-1,1)^n)\ f(x) = 0$, by
 Proposition~\ref{prop:boundedopen}. Then $S = \{(x,y) \in (-1,1)^n \times \RN: f(x) = 0\}$ is unbounded if and only if $\Phi$ is true.
\end{proof}

Our proof shows the slightly stronger result that testing unboundedness of an algebraic set is \NPR-complete.

We say that a semialgebraic set $S$ has {\em diameter at most $d$} if $(\forall x\in S)(\forall y\in S)\ \lVert x-y \rVert \leq d$. The proof of Theorem~\ref{thm:unbounded} shows that testing whether a semialgebraic (indeed algebraic) set has diameter at most $1$, is \coNPR-complete. So what about radius? We say that $S$ has {\em radius at most $r$} if $(\exists x)(\forall y)\ y \in S \implies \lVert x-y \rVert \leq r$. By definition, the radius problem lies in \SR{2}{}. Somewhat surprisingly, it belongs to \coNPR.

\begin{theorem}
  Testing whether the radius of a semialgebraic set is at most $r$ is \coNPR-complete.
\end{theorem}

\begin{proof}
Let $S\subseteq\RN^d$ be a semialgebraic set defined by a predicate $\varphi(x)$, that is, $S=\{x\in\RN^d: \varphi(x)\}$. \coNPR-hardness of the problem follows as in the proof of Theorem~\ref{thm:unbounded}. We need to show that the problem lies in \coNPR.

The statement ``$S$ has radius at most $r$'' can be expressed as
\begin{equation}\label{eq:rad1}
(\exists x\in\RN^d) (\forall y\in\RN^d)\ \varphi(y) \implies \lVert x-y\rVert^2 \leq r^2.
\end{equation}
By Helly's theorem,~\eqref{eq:rad1} is equivalent to
\begin{equation}\label{eq:rad2}
(\forall y_1,\dots,y_{d+1}\in\RN^d) (\exists x\in\RN^d)\ \bigwedge_{k\in [d+1]} \varphi(y_k) \implies \bigwedge_{k\in [d+1]} \lVert x-y_k\rVert^2 \leq r^2,
\end{equation}
where $[d+1] = \{1,\ldots, d+1\}$.
Moving the existential quantifier, we obtain the equivalent formula
\begin{equation}\label{eq:rad3}
(\forall y_1,\dots,y_{d+1}\in\RN^d)\ \bigwedge_{k\in [d+1]} \varphi(y_k) \implies (\exists x\in\RN^d)  \bigwedge_{k\in [d+1]} \lVert x-y_k\rVert^2 \leq r^2 .
\end{equation}

The following claim allows us to replace the inner existential quantifier block with a universal quantifier block completing the proof.

\begin{claim}
The following two statements are equivalent:
\begin{equation}\label{eq:rad4}
(\exists x\in\RN^d)  \bigwedge_{k\in [d+1]} \lVert x-y_k\rVert^2 \leq r^2,
\end{equation}
and
\begin{equation}\label{eq:rad6}
(\forall z\in\RN^{d+1}_{\geq 0})\  \sum_{k=1}^{d+1} z_k = 1 \implies \sum_{k=1}^{d+1} \sum_{l=1}^{d+1} z_k z_\ell \lVert y_k-y_{\ell}\rVert^2 \leq 2r^2.
\end{equation}
\end{claim}

We are left with the proof of the claim. The minimum radius of a sphere enclosing all $y_k$ is given by the following program (the value of the program is the square of the radius).
\begin{eqnarray*}
\mbox{minimize} && z \\
\mbox{subject to} && x^T x - 2x^T y_k + y_k^T y_k\leq z,\quad k\in [d+1].
\end{eqnarray*}
Let $w=z-x^Tx$. The program becomes
\begin{eqnarray}
\mbox{minimize} && w + x^T x \label{eq:primal} \\
\mbox{subject to} && w + 2x^T y_k \geq y_k^T y_k,\quad k\in [d+1]. \nonumber
\end{eqnarray}
The Lagrangian for the program is
\[
L(w,x,z) = w + x^T x - \sum_{k=0}^d z_k (w + 2x^T y_k - y_k^T y_k).
\]
We have
\[
\frac{\partial}{\partial x} L(w,x,z) = 2 ( x - \sum_{k=1}^{d+1} z_k y_k )
\]
and
\[
\frac{\partial}{\partial w} L(w,x,z) = 1 - \sum_{k=1}^{d+1} z_k.
\]
It follows that the critical points satisfy
\[
x = \sum_{k=1}^{d+1} z_k y_k\quad\mbox{and}\quad \sum_{k=1}^{d+1} z_k = 1,
\]
and we have
\[
L(w,y,z) = \sum_{k=1}^{d+1} z_k (y_k^T y_k) - \sum_{k=1}^{d+1} \sum_{\ell=1}^{d+1} z_k z_\ell y_k^T y_\ell.
\]

With this, we can write the dual program as
\begin{eqnarray}
\mbox{maximize} && \sum_{k=1}^{d+1} z_k (y_k^T y_k) - \sum_{k=1}^{d+1} \sum_{\ell=1}^{d+1} z_k z_\ell y_k^T y_\ell \label{eq:dual} \\
\mbox{subject to} && \sum_{k=1}^{d+1} z_k = 1, z_k\geq 0\ \mbox{for}\ k\in [d+1]. \nonumber
\end{eqnarray}

The primal program~\eqref{eq:primal} is strictly feasible (take sufficiently large $w$) and hence, by Slater's condition, we have
strong duality, see~\cite[p. 226]{BV04}. This means that~\eqref{eq:primal}
has value at most $r^2$ if every feasible solution of~\eqref{eq:dual} has value at most $r^2$. This is equivalent to
\begin{equation}\label{eq:rad5}
(\forall z\in\RN^{d+1}_{\geq 0})\  \sum_{k=1}^{d+1} z_k = 1 \implies \sum_{k=1}^{d+1} z_k (y_k^T y_k) - \sum_{k=1}^{d+1} \sum_{\ell=1}^{d+1} z_k z_\ell y_k^T y_\ell\leq r^2.
\end{equation}
Since
\begin{equation*}
\sum_{k=1}^{d+1} z_k (y_k^T y_k) - \sum_{k=1}^{d+1} \sum_{\ell=1}^{d+1} z_k z_\ell y_k^T y_\ell = \sum_{k=1}^{d+1} \sum_{\ell=1}^{d+1} z_k z_\ell (y_k^T y_k - y_k^T y_\ell),
\end{equation*}
using $\sum_{\ell=1}^{d+1} z_{\ell} = 1$, we can simplify~\eqref{eq:rad5} to obtain~\eqref{eq:rad6}.
\end{proof}

Another tricky problem is testing whether a semialgebraic set $S$ is dense in another set $T$. Let us assume that $T$ is the ambient space $\RN^n$, we will discuss more general choices of $T$ in the next section. The denseness of $S = \{x : \phi(x)\}$ in $\RN^n$ can be expressed as
\[(\forall x)(\forall \varepsilon > 0)(\exists x')\ \lVert x-x'\rVert < \varepsilon \wedge \phi(x'),\]
which is of the form $\forall\exists$. Very surprisingly, Koiran~\cite{K99} showed that the existential quantifier can be eliminated in this case\footnote{As the reviewer points out, Koiran's result is even stronger, it applies in the BSS-model.}, and the problem lies in \coNPR. Compared to that, proving \coNPR-hardness is relatively easy. The following result was proved for semialgebraic sets defined by algebraic circuits rather than formulas by B{\"u}rgisser and Cucker~\cite[Proposition 5.2, Corollary 9.3$(ii)$]{BC09}. They discuss the difficulty of moving from circuits to formulas for this problem in their Remark 5.4$(ii)$.

\begin{theorem}\label{thm:semidense}
  Testing whether a semialgebraic set $S \subseteq \RN^n$ is dense in $\RN^n$ is \coNPR-complete.
\end{theorem}

As the proof will show, the ambient space $\RN^n$ is not special, it can easily be replaced with other semialgebraic sets such as $[-1,1]^n$.

\begin{proof}
 Membership in \coNPR\ is non-trivial and follows from work by Koiran~\cite{K99}, also see~\cite[Corollary 5.3]{BC09}. To prove hardness, we work with an instance $\Phi = (\forall x \in [-1,1]^n)\ f(x) > 0$, which is \coNPR-complete by~\cite[Theorem 7]{OW17}. If $\Phi$ is true, then Theorem~\ref{thm:Linq}, Solern\'{o}'s inequality, implies, just as in the proof of \ref{cor:strictify2}, that
 \[
\vert f(x)\vert \geq 2^{-\ell D^{cn^2}}
\]
for all $x \in [-1,1]^n$, where $D$ is the degree of $f$ and $\ell$ bounds the number of bits in every coefficient of $f$.
With $m = \lceil \ell D^{cn^2} \rceil+1$ define the semialgebraic set
\begin{equation*}
\begin{split}
 T = \{y \in [-1,1]^m : \; & 1/4 \leq y_1 \leq 1/2,  \\
                        & y_1^2/2 \leq y_2 \leq y_1^2, \ldots, y_{m-1}^2/2 \leq y_m \leq y_{m-1}^2\}.
\end{split}
\end{equation*}
For $y \in T$ we have $0 < 4^{-\ell D^{cn^2}} < y_m < 2^{-\ell D^{cn^2}}$, so $f(x) > 0$ implies that $f(x) > y_m$ for all $y \in T$. We can then define
\begin{equation*}
\begin{split}
 S = & \{x \in [-1,1]^n, y \in [-1,1]^m :  f(x) \geq y_m \wedge y \in T\} \\
     & \cup (\RN^{n+m}-[-1,1]^{n}\times ([-1,1]^m-T)).
\end{split}
\end{equation*}
By definition of $T$ and $S$, the truth of $\Phi$ implies that $S = \RN^{n+m}$, in particular, $S$ is dense in $\RN^{n+m}$. Assume then that $\Phi$ is false, so there is an $x \in [-1,1]^n$ for which $f(x) = 0$. In this case $(x,y) \not\in S$ for any $y \in T$. Since $f$ is continuous, we can pick a $\varepsilon > 0$ such that $\vert f(x')\vert  < 4^{-\ell D^{cn^2}}$ for all $x'$ with $\lVert x-x' \rVert < \varepsilon$.
Since $y_m > 4^{-\ell D^{cn^2}}$, this implies that $(x',y) \not\in S$ for all such $x'$ and all $y \in T$, so $S$ is not dense in $\RN^{n+m}$.
\end{proof}

Following B\"{u}rgisser and Cucker, we can capture the complexity of the denseness problem using the quantifier $\forall^*$, where $(\forall^* x)\ \varphi(x)$ is defined as
\[(\forall x)(\forall \varepsilon > 0)(\exists x')\ \lVert x-x'\rVert < \varepsilon \wedge \varphi(x').\]
The $\forall^*$ quantifier can be read as ``for almost all'', or ``for a dense set of''. The meaning of the dual quantifier $\exists^*$ in $(\exists^* x)\ \varphi(x)$ is defined by
\[(\exists x)(\exists \varepsilon > 0)(\forall x')\ \lVert x-x'\rVert < \varepsilon \implies \varphi(x').\]
So $\forall^* = \neg \exists^* \neg$.

Using $\forall^*$ the denseness of the semialgebraic set $S = \{x : \varphi(x)\}$ can be written simply as $(\forall^* x)\ \phi(x)$. B\"{u}rgisser and Cucker introduce the classes $\BP^0(\exists^*)$ and $\BP^0(\forall^*)$ for problems of this type, but there is one subtle difference to our approach: since they base their approach on the BSS machine model, they need to allow algebraic circuits to define semialgebraic sets whenever the final quantifier block is exotic, that is $\exists^*$, $\forall^*$, or $H$. So the complete problem for $\BP^0(\forall^*)$ is of the form $(\forall^* x)\ C(x)$, where $C$ is an algebraic circuit, see~\cite[Section 2]{BC09}. This leads to stronger upper bound results, but weaker lower bounds. At the first level, the exotic classes collapse to their non-exotic counterparts.

\begin{corollary}[B\"{u}rgisser and Cucker~{\protect\cite[Corollary 9.3$(ii)$]{BC09}}]\label{cor:forallstarforall}
 $\BP^0(\exists^*) = \BP^0(\exists) = \NPR$ and $\BP^0(\forall^*) = \BP^0(\forall) = \coNPR$.
\end{corollary}

B\"{u}rgisser and Cucker show several additional, non-trivial relationships between the exotic and the standard quantifiers~\cite[Corollary 6.2, Corollary 8.3, Proposition 8.4]{BC09}.

If the final quantifier is $H$, we can easily remove it using quantifier elimination.

\begin{lemma}\label{lem:remfinalH}
    Let $\Phi$ be the sentence $(Q_1x_1)\cdots(Q_ix_i)(H \varepsilon)\ \varphi(x_1, \ldots, x_i,\varepsilon)$, where $Q_i \in \{\exists,\forall,\exists^*,\forall^*,H\}$. We can construct, in polynomial time, an equivalent sentence $\Psi =  (Q_1x_1)\cdots(Q_ix_i)\ \psi(x_1, \ldots, x_i)$.
\end{lemma}
\begin{proof}
 Consider $\varphi(x_1, \ldots, x_i,\varepsilon)$ as a formula in the single variable $\varepsilon$. Using the definition of $H$ we can rewrite the sentence $(H \varepsilon)\ \varphi(x_1, \ldots, x_i,\varepsilon)$ as $(\exists \varepsilon' > 0)(\forall \varepsilon \in (0,\varepsilon'))\ \varphi(x_1, \ldots, x_i,\varepsilon)$, and then apply quantifier elimination (e.g.\ in the form of Theorem~1 in~\cite{S91}) twice to eliminate the two quantifiers. This gives us, in polynomial time, a formula $\psi(x_1, \ldots, x_i)$ equivalent to $(H \varepsilon)\ \varphi(x_1, \ldots, x_i,\varepsilon)$. Reintroducing the quantifiers $Q_i$, $Q_{i-1}$, $\ldots$, $Q_1$, in this order, we obtain that $\Phi$ is equivalent to $\Psi = (Q_1x_1)\cdots(Q_ix_i)\ \psi(x_1, \ldots, x_i)$.
\end{proof}

Applying the lemma repeatedly, allows us to remove a fixed number of final $H$ quantifiers. We would like to conclude that $\BP^0(\exists^* H) = \BP^0(\exists H) = \NPR$ and $\BP^0(\forall^* H) = \BP^0(\forall H) = \coNPR$, but we cannot, because the complete problems for any of the $H$ classes involve algebraic circuits; to rewrite these as formulas, we would have to add existential quantifiers after the $H$, leading to formulas of the form $\exists^* H \exists$. We cannot then apply quantifier elimination to $H$ without dealing with the final block of existential quantifiers first.

What can we do?  B\"{u}rgisser and Cucker studied the complexity of the \LS\ problem: they showed that testing whether a given hyperplane locally supports a semialgebraic set (given by an algebraic circuit) is complete for $\BP^0(\exists^* H)$. We can prove an analogue for semialgebraic sets defined by formulas.

A hyperplane {\em locally supports} a set $T$ if there is a point on the hyperplane so that a neighborhood of that point in the hyperplane belongs to $T$, and $T$, at least close to the point, lies entirely on one side of the hyperplane. For example, a square in $\RN^2$ has four supporting hyperplanes (lines); a line just passing through a corner is not considered locally supporting in this definition, since its intersection with the square is not two-dimensional.

\begin{corollary}
  Given a hyperplane and a bounded semialgebraic set, testing whether the hyperplane locally supports the semialgebraic set is \NPR-complete.
\end{corollary}

\begin{proof}
For membership we follow B\"{u}rgisser and Cucker~\cite[Proposition 7.2]{BC09}.
 We can assume the hyperplane $S$ is given as $S = \{x: v \cdot x = c\}$ with $v \in \QN^n$ and $c \in \QN$. Let $h$ be the affine function for which $S = \{(x,h(x)): x \in \RN^{n-1}\}$ and let $T = \{x: \varphi(x)\}$ be the semialgebraic sets. Then $S$ locally supports $T$ if and only if
 \[(\exists^* x \in \RN^{n-1}) (H \varepsilon)\ [\varphi((x,h(x))+\varepsilon v) \wedge \neg\varphi((x,h(x))-\varepsilon v) ) ].   \]
Using Lemma~\ref{lem:remfinalH} we can eliminate $H$, so we obtain a formula $\psi(x)$ such that
$S$ locally supports $T$ if and only if $(\exists^* x)\ \psi(x)$, but then the problem lies in \NPR\ by Corollary~\ref{cor:forallstarforall}.

To show \NPR-hardness, we can work with the sentence
\[\Phi = (\exists x \in (-1,1)^i)\ f(x) > 0,\]
 by Proposition~\ref{prop:boundedopen}. Let $S = \{(x,y) \in \RN^i \times \RN: y = 0\}$ and $T = \{(x,y) \in \RN^i\times \RN: f(x) > 0 \wedge y \geq 0\}$. Then $S$ is a hyperplane, and $T$ a bounded semialgebraic set. $S$ touches $T$ if $S \cap T$ is $i$-dimensional. But that is the case if and only of there is an $x \in \RN^i$ such that $f(x)> 0$, since if $f(x)> 0$ is true for some $x$, it is true for a neighborhood of $x$, by continuity.
\end{proof}

For a list of further problems complete for \NPR\ and \coNPR\ see~\cite[Corollary 9.4]{BC09}.

\subsection{Second-Level Problems}

B\"{u}rgisser and Cucker~\cite[Proposition 4.1, Table~1]{BC09} showed that testing whether a piecewise rational function given by an algebraic circuit is surjective is \PR{2}{}-complete. Below we include a proof that extends their result to polynomials (of bounded degree).

\begin{theorem}\label{thm:surjective}
 Testing whether a polynomial $g: A \rightarrow B$ is surjective is \PR{2}{}-complete, even if $g$ is a polynomial of degree at most $8$, and $A$ and $B$ are Cartesian products of $\RN$, and $[-1,1]$.
\end{theorem}

\begin{proof}
 By Proposition~\ref{prop:bc2nd}, \bd\PPOLY{2}{=}\ is \PR{2}{}-complete. So we can assume we are given a sentence  $\Phi = (\forall x \in [-1,1]^n)(\exists y \in [-1,1]^m)\ f(x,y) = 0$ with an explicit polynomial $f$ of degree at most $4$. Define a new polynomial
 \[g: \begin{array}{cll}
         [-1,1]^n \times [-1,1]^m \times \RN \times \RN & \rightarrow & [-1,1]^n \times \RN \\
         (x,y,s,t) & \mapsto & (x, s - (sf(x,y))^2-t^2)
      \end{array} \]
 Then $g$ is surjective if and only of $\Phi$ is true: If $\Phi$ is true, then for every $x \in [-1,1]^n$ there is a $y \in [-1,1]^m$ such that $f(x,y) = 0$ and $g(x,y,s,0) = (x,s)$, which shows that $g$ is surjective. If $\Phi$ is false, then there is an $x \in [-1,1]^n$ such that $(f(x,y))^2 > 0$ for all $y \in [-1,1]^m$. By compactness, there is a $\varepsilon$ with $(f(x,y))^2 > \varepsilon > 0$ for all $y \in [-1,1]^m$. Then $s - (sf(x,y))^2-t^2 < s- (s\varepsilon)^2$, which has a finite upper bound for $(s,t) \in \RN^2$, so $g$ is not surjective.
\end{proof}

There are various ways to measure how close two semialgebraic sets are. For example, testing equality of two semialgebraic sets is \coNPR-complete~\cite[Theorem 3]{JKM22}, and the same is true for containment: $S \subseteq T$.
What happens if we replace equality with denseness? In the previous section we looked at the problem of deciding whether a set $S$ contained in a set $T$ is dense in $T$ (for $T = \RN^n$). Writing $\overline{S}$ for the closure of the set $S$ we can express this as $\overline{S} = T$.

What about $S \subseteq \overline{T}$ instead? If $S = \{x: \varphi(x)\}$ and $T = \{x: \psi(x)\}$, then $S \subseteq \overline{T}$ is equivalent to
\[(\forall^* x)(\forall^* \varepsilon)(\exists y)\ \varphi(x) \wedge \psi(y) \wedge \lVert x-y \rVert < \varepsilon,\]
and B{\"u}rgisser and Cucker~\cite[Proposition 5.5]{BC09} show completeness of testing $S \subseteq \overline{T}$ for problems of the type $\forall^*\exists$. In the discrete version, their result implies completeness for the class $\BP^0(\forall^*\exists)$. Clearly, $\BP^0(\forall^*\exists) \subseteq \BP^0(\forall\exists)$.\footnote{At the second level we do not need Koiran's method for quantifier elimination~\cite{K99}, we can argue directly: $\forall^*$ can be rewritten as $\forall\exists$, and the additional existential quantifiers merged with the existing ones.}
We argue below that the two classes are the same. As was the case for the denseness problem at the first level, the quantifier $\forall^*$ does not differ from $\forall$ in computational power at the second level.

Our proof that testing $S \subseteq \overline{T}$ is \PR{2}{}-complete is
closely based on~\cite[Proposition 5.5]{BC09}.

\begin{theorem}\label{thm:containPR2}
 Testing whether $S \subseteq \overline{T}$ for two semialgebraic sets $S$ and $T$ is \PR{2}{}-complete, even if $S$ is of the form $[-1,1]^{p} \times \{0\}^q$.
\end{theorem}
\begin{proof}
 Since the $\forall^*$-quantifier can be replaced by $\forall\exists$, membership in \PR{2}{} easily follows. For the other direction, we work with Proposition~\ref{prop:bc2nd}. So we are given a sentence
 \[\Phi = (\forall x \in [-1,1]^n)(\exists y \in [-1,1]^m)\ f(x,y) = 0.\]

 Write $f(x,y) = \sum_{\alpha} f_{\alpha}(x) y^{\alpha}$. Following~\cite{BC09} we define
 \[f'(x,y,y_0) = \sum_{\alpha} f_{\alpha}(x) y_0^{d-\vert \alpha\vert } y^{\alpha} = y_0^d \sum_{\alpha} f_{\alpha}(x) (y/y_0)^{\alpha} = y_0^d f(x,y/y_0),\]
 the homogenization of $f$, where $d$ is the largest total degree of $y$ in $f$.

 Define the sets
 \begin{eqnarray*}
  S &=& [-1,1]^n\times \{0\}^{m} \times \{0\},\ \mbox{and} \\
  T &=& \{(x,y,y_0)\in [-1,1]^n\times [-1,1]^m \times [0,1]: \\
    & & \quad\quad\quad f'(x,y,y_0) = 0 \wedge y_0 \neq 0 \wedge \lVert y \rVert \leq y_0\}.
 \end{eqnarray*}
 We claim that $\Phi$ is true if and only if $S \subseteq \overline{T}$.

 Suppose $\Phi$ is true, and let $(x,0,0) \in S$. Since $\Phi$ is true, there is a $y$ such that $f(x,y) = 0$. We can then find a point $(x, y', y_0)$ in $T$ arbitrarily close to $(x,0,0)$ by choosing $y_0$ arbitrarily small (but bigger than $0$), and setting $y' = y_0y$. Then
 \[f'(x, y', y_0) = y_0^d f(x,y) = 0,\]
 and $\lVert y' \rVert \leq y_0 \lVert y \rVert \leq y_0$, so $(x,0,0)$ lies in $\overline{T}$.

 For the other direction, the difference between the BSS-model, and the discrete model comes into play. Suppose $S \subseteq \overline{T}$, and we are given $x \in [-1,1]^n$. We have to show that there is a $y \in [-1,1]^m$ such that $f(x,y) = 0$.

 For a contradiction assume that $f(x,y) \neq 0$ for all $y \in [-1,1]$; without loss of generality, $f(x,y) > 0$ for all $y \in [-1,1]^m$ (if not, replace $f$ with $-f$, we use the fact that $f$ is continuous).
 Since $[-1,1]^m$ is compact, this implies that the minimum of $f(x,y)$ over all
 $y \in [-1,1]^m$ is greater than some $\delta > 0$. Since $f$ is continuous (as a polynomial), this remains true in a sufficiently small neighborhood of $x$. But then
 \begin{eqnarray*}
 f'(x', y', y_0) & = &  y_0^d \sum_{\alpha} f_{\alpha}(x) (y'/y_0)^{\alpha} \\
                 & \geq & y_0^d  \delta.
 \end{eqnarray*}
 for all $x'$ sufficiently close to $x$, and all $y'$ with $\lVert y' \rVert \leq y_0$. In particular, $f'(x', y', y_0) > 0$ for $y_0 \neq 0$, so $(x', y', y_0) \not\in T$. Hence, $(x,0,0)$ does not belong to $\overline{T}$. This contradicts $(x,0,0) \in S \subseteq \overline{T}$.
\end{proof}

The sets $S$ and $T$ in Theorem~\ref{thm:containPR2} are basic semialgebraic sets. We do not know whether the hardness result can be extended to $S$ and $T$ being algebraic.

\begin{corollary}\label{cor:forallstar2eqforall2}
    $\BP^0(\forall^*\exists) = \BP^0(\forall\exists) = \PR{2}{}$.
\end{corollary}

The corollary also settles the complexity of the problems \ERD, \LERD, and \IED\ mentioned in~\cite{BC09}, when restricted to their Boolean parts. They are all \PR{2}{}-complete. With the exception of \LS, these were the last remaining problems from~\cite{BC09} whose complexity in terms of the real hierarchy remained open because they included exotic quantifiers. We have to exclude \LS, since the complexity of that problem remains open for semialgebraic sets defined by circuits.

As a second consequence of Theorem~\ref{thm:containPR2}, we can show that the Hausdorff distance problem remains \PR{2}{}-complete for distance $0$. Jungeblut, Kleist and Miltzow, in the second version of~\cite{JKM22}, also realized that B\"{u}rgisser and Cucker's proof of Proposition 5.5 in~\cite{BC09} implies that the problem is \PR{2}{<}-complete.

\begin{corollary}\label{cor:Hausdorff}
 Given two semialgebraic sets $S$ and $T$ testing whether their (directed) Hausdorff distance is zero is \PR{2}{}-complete.
\end{corollary}

\begin{proof}
 The directed Hausdorff distance of $S$ and $T$ is $0$ if and only if $S \subseteq \overline{T}$, so hardness of the directed case follows from Theorem~\ref{thm:containPR2}.

 The directed case reduces to the undirected case as follows: $S \subseteq \overline{T}$ is equivalent to $S \cup T \subseteq \overline{T}$. Since $\overline{T} \subseteq \overline{S \cup T}$ this implies that $S
 \subseteq \overline{T}$ if and only if $S\cup T$ and $T$ have the same closure: $\overline{S \cup T} = \overline{T}$, which is equivalent to $S \cup T$ and $T$ having Hausdorff distance $0$.
\end{proof}

\section{Conclusion and Open Problems}

In this paper we (re)introduced the real hierarchy, a hierarchy of complexity classes based on the theory of the reals, and justified the definition by showing that, just like \NPR, it is robust under changes in the definition. We identified several, very restricted families of problems complete for \SR{k}{}, see Table~\ref{tab:shcomp}, and \PR{k}{}, see Table~\ref{tab:phcomp}.

\begin{table}[htb]
\caption{Degree bounds and references for complete problems for \protect\SR{k}{}; values in green are improved compared to the journal version. The bounded-closed domain cases have been removed.}\label{tab:shcomp}
\begin{tabular}{l|l|l}
     \SR{k}{}-complete         &  $k$ even                                              & $k$ odd \\ \hline
 \SPOLY{k}{<} & degree $\color{Green} 4$ (P\ref{prop:QPOLY})  & degree $\color{Green} 6$ (P\ref{prop:strictcomp}) \\
 \SPOLY{k}{=} &                     na                &  degree $4$ (P\ref{prop:QPOLY}) \\
 \SPOLY{k}{\leq} & degree $\color{Green} 6$ (P\ref{prop:strictcomp} + C\ref{cor:leq}) & degree $\color{Green} 4$ (P\ref{prop:QPOLY} + C\ref{cor:leq}) \\
 \bo\SPOLY{k}{<} & degree $\color{Green} 8$ (P\ref{prop:boundedopen})  & degree $\color{Green} 4$ (P\ref{prop:boundedopen}) \\
 %\bd\SPOLY{k}{<} & degree $\color{Green} 8$ (P\ref{prop:boundedopen})  & degree $\color{Green} 4$ (P\ref{prop:boundedopen}) \\
 \bo\SPOLY{k}{=} &                     na                &  degree $4$ (P\ref{prop:boundedopen}) \\
 %bd\SPOLY{k}{=} &                     na                &  degree $4$ (C\ref{cor:eqbounded}) \\
 \end{tabular}
 \end{table}

\begin{table}[htb]
\caption{Degree bounds and references for complete problems for \protect\PR{k}{}; values in green are improved compared to the journal version. The bounded-closed domain cases have been removed.}\label{tab:phcomp}
 \begin{tabular}{l|l|l}
  \PR{k}{}-complete  &  $k$ even                                              & $k$ odd \\ \hline
  \PPOLY{k}{<} & degree $\color{Green} 6$ (P\ref{prop:strictcomp}) & degree $\color{Green} 4$ (P\ref{prop:QPOLY}) \\
  \PPOLY{k}{=} & degree $4$ (P\ref{prop:QPOLY}) &  na  \\
  \PPOLY{k}{\leq} & degree $\color{Green} 4$  (P\ref{prop:QPOLY} + C\ref{cor:leq})  & degree $\color{Green} 6$ (P\ref{prop:strictcomp}+ C\ref{cor:leq}) \\
 \bo\PPOLY{k}{<} & degree $\color{Green} 4$ (P\ref{prop:boundedopen}) & degree $\color{Green} 8$ (P\ref{prop:boundedopen}) \\
 %\bd\PPOLY{k}{<} & degree $8$ (C\ref{cor:strictify2})  & degree $\color{Green} 8$ (P\ref{prop:boundedopen}) \\
 \bo\PPOLY{k}{=} & degree $4$ (P\ref{prop:boundedopen})  & na \\
 %\bd\PPOLY{k}{=} & degree $4$ (C\ref{cor:eqbounded}) &  na\\
\end{tabular}
\end{table}

We emphasize that the bounded-closed domain versions which were listed in the journal version have been removed from the paper, since their proofs did not hold up (except at the first and second level); we are left with Conjecture~\ref{con:bceq} which we already stated earlier.

We hope that these problems will be useful in future investigations at higher levels of the real hierarchy. Are there any natural candidates for complete problems at higher levels? Dobbins, Kleist, Miltzow and Rz\k{a}\.{z}ewski~\cite{DKMR22} explore this question in depth for \SR{2}{}\ and \PR{2}{}. They identify three mechanisms that often lead to a jump in complexity:
universal extension, imprecision, and robustness. In a {\em universal extension} variant, we ask whether any given partial solution can be extended to a full solution. E.g.\ starting with the \NP-complete graph coloring, the problem \twocolext, whether any partial $2$-coloring of a given set of vertices can be extended to a $3$-coloring of the graph is $\coNP^{\NP}$-complete~\cite{S05}; so the complexity makes a $\coNP$-jump. Dobbins \etal\ discuss a universal extension variant of the art gallery problem, and also mention some other candidates. In an {\em imprecision} variant each real parameter is replaced by a (metric) neighborhood. E.g.\ the coordinates of the art gallery may be given to within some precision bound. Can all the resulting art galleries be guarded by at most $k$ guards~\cite{DKMR22}? The {\em robustness} variant applies the precision bound to the solution set: Given an art gallery are there $k$ guards that guard the whole gallery, even if they are perturbed (each within some precision bound)~\cite{DKMR22}? We refer to the paper by Dobbins \etal~\cite{DKMR22} for more details, discussion, and further candidates.

As we saw, other natural problems arise when studying properties of semialgebraic sets, such as the Hausdorff distance of two semialgebraic sets. The paper by B{\"u}rgisser and Cucker~\cite{BC09} includes a pretty comprehensive look at computational questions about semialgebraic sets, and it leaves some interesting open questions. For example, they were able to show that testing whether a basic semialgebraic set is closed is \coNPR-complete, but it is open whether membership in \coNPR\ still holds for semialgebraic sets, see~\cite[Theorem 6.15]{BC09}. A problem of practical interest~\cite{LSdW20} is testing whether a semialgebraic set contains an isolated point. The problem is \coNPR-hard~\cite[Corollary 9.4]{BC09}, but the best known upper bound is \SR{2}{}. Some further properties to consider: A set is {\em star-shaped} if it contains a point that can ``see'' all other points in the set; by definition, testing whether a semialgebraic set is star-shaped belongs to \SR{2}{}, but is it \SR{2}{}-complete?

We can also consider replacing families of geometric objects with parameterized families of semialgebraic sets. In this context one could study the complexity of computing the Vapnik-\v{C}ervonenkis dimension of the family $\{x \in \RN^n: (\exists y \in \RN^m)\ \varphi(x,y)\}$. The problem clearly lies in \SR{3}{}. Is it \SR{3}{}-complete? This appears unlikely, since the universal quantifier is not really real, but Boolean, quantifying over subsets, so the problem lies in a hybrid discrete/real complexity class, another topic that deserves attention.\footnote{We'd like to write $\NPR^{\coNP^{\NPR}}$ for this class, but this notation suggests an oracle model for \NPR, and the details of that would still need to be worked out.} Blanc and Hansen~\cite{BH22a} study a problem in evolutionary game theory which they can show lies in \SR{2}{}, and is hard for the hybrid class that is of the form $\exists\forall$, where the existential quantifier is Boolean.

Other problems related to semialgebraic sets, like testing connectivity, or counting connected components likely do not lie at a finite level of the real hierarchy, since they cannot be defined in first-order logic, but it would be interesting to find out for which levels they are hard, following Basu and Zell~\cite{BZ10}.

Finally, we need a deeper understanding of the power of B{\"u}rgisser and Cucker's exotic quantifers in the discrete setting. We conjecture that $H$ does not affect the complexity at all, while $\forall^*$ and $\exists^*$ can be replaced with their non-exotic counterparts (at least in most situations). We have only shown this for $\forall^*$ at the second level, and we saw that $H$ can be eliminated as a first and last quantifier, but what if $H$ occurs in the middle? We are not aware of any natural problems that would require such an $H$-quantifier, with the exception of the \LS-problem which has this flavor if we use existential quantifiers to replace the algebraic circuit defining the semialgebraic set with a formula.

Let us move on to some more structural, complexity-theoretic questions. Renegar's quantifier elimination result implies that formulas with $O(\log n)$ quantifier alternations can be decided in \PSPACE. Is it possible that the logarithmic fragment of the theory of the reals exhausts all of \PSPACE? We would also like to see some oracle separations. The theory of the reals can be relativized in various (standard) ways. Are there oracles that separate \NP\ from \NPR? Or \NPR\ from \PSPACE? Because $\P \subseteq \NPR \subseteq \PSPACE \subseteq \EXP$, any oracle that collapses \PSPACE\ to \P\ separates \NPR\ from \EXP. On the other hand, since $\NP \subseteq \NPR$, any oracle that separates $\P$ and $\NP$ also separates $\P$ and $\NPR$, but these few facts seem to be the extent of our knowledge. In particular, we do not know whether the real hierarchy is proper (relative to an oracle).

Are there interesting problems in $\NPR \cap \coNPR$? One such problem is semidefinite programming; we can phrase the decision version of semidefinite programming as follows~\cite{L03}: Given symmetric $n\times n$ matrices $A_1, \ldots, A_m$, $B$ are there $x_1, \ldots, x_m \in \RN$ such that $\sum_{i \in [m]} x_iA_i - B$ is positive semidefinite? Ramana~\cite{R97} showed that the problem lies in $\NP_{\RN} \cap \coNP_{\RN}$, that is the intersection of the real versions of \NP\ and \coNP\ in the BSS-model, and it is not known whether the problem lies in $\NP \cap \coNP$. Using Ramana's duality result for semidefinite programming it is easy to show that the problem also lies in $\NPR \cap \coNPR$. One reason this is interesting is that the sum of square roots problem, that is, deciding the truth of sentences of the form $\sum_{i \in [n]} \sqrt{a_i} \leq k$, is a special case of semidefinite programming (e.g.~\cite{G98}), and is one of the few problems in \NPR\ for which a better complexity bound than \PSPACE\ is known; it is located in the counting hierarchy~\cite{ABKPM08}.

Finally, we should mention that Erickson, van der Hoog, and Miltzow~\cite{EvdHM20} defined a RAM-machine model for \NPR\ which goes beyond the BSS-model to capture integer operations; this model was extended to $\PR{2}{}$ by Dobbins, Kleist, Miltzow and Rz\k{a}\.{z}ewski~\cite{DKMR22}, and it appears likely that it can be generalized to arbitrary levels, maybe even including an oracle mechanism.

\subsection*{Acknowledgments}

We would like to thank an anonymous reader who suggested the paper by B\"{u}rgisser and Cucker~\cite{BC09}, which led to the addition of the results in Section~\ref{sec:SAS}. We would also like to thank Kristoffer Arnsfelt Hansen and Soumyajit Paul who gave valuable feedback on the paper.

\bibliographystyle{plain}
\bibliography{BeyondNPR}

\end{document}